\documentclass[preprint]{article}
\usepackage[top=1in, bottom=1in, left=0.5in, right=0.5in]{geometry}
\usepackage{fancyhdr}

\usepackage{adjustbox}
\usepackage{amsmath}
\usepackage{amsfonts}
\usepackage{amssymb}
\usepackage{bm}
\usepackage[table]{xcolor}
\usepackage{tabu}
\usepackage{makecell}
\usepackage{longtable}
\usepackage{multirow}
\usepackage[normalem]{ulem}
\usepackage{etoolbox}
\usepackage{graphicx}
\usepackage{tabularx}
\usepackage{ragged2e}
\usepackage{booktabs}
\usepackage{caption}
\usepackage{fixltx2e}
\usepackage{threeparttablex}
\usepackage[capposition=top]{floatrow}
\usepackage{subcaption}
\usepackage{pdfpages}
\usepackage{pdflscape}
\usepackage{natbib}
\usepackage{bibunits}
\definecolor{maroon}{HTML}{800000}
\usepackage[colorlinks=true,linkcolor=maroon,citecolor=black,urlcolor=maroon,anchorcolor=maroon]{hyperref}
\usepackage{marvosym}
\usepackage{makeidx}
\usepackage{tikz}
\usetikzlibrary{shapes}
\usepackage{setspace}
\usepackage{enumerate}
\usepackage{rotating}
\usepackage{epstopdf}
\usepackage[toc,page]{appendix}
\usepackage{framed}
\usepackage{comment}
\usepackage{xr}
\usepackage{mathtools}
\usepackage{dsfont}
\usepackage{xfrac}
\usepackage{nicefrac}
\usepackage{amsthm}
\usepackage{lipsum} 

\usepackage{footnote}
\usepackage{longtable}
\newlength{\tablewidth}
\makeatletter
\pretocmd\start@align
{%
  \let\everycr\CT@everycr
  \CT@start
}{}{}
\apptocmd{\endalign}{\CT@end}{}{}
\makeatother
\usepackage[printwatermark]{xwatermark}
\usepackage{lipsum}
\definecolor{lightgray}{RGB}{220,220,220}
\usepackage{titlesec}
\titleformat{\paragraph}
{\normalfont\normalsize\bfseries}{\theparagraph}{1em}{}
\titlespacing*{\paragraph}
{0pt}{3.25ex plus 1ex minus .2ex}{1.5ex plus .2ex}

\titleformat{\subparagraph}
{\normalfont\normalsize\bfseries}{\thesubparagraph}{1em}{}
\titlespacing*{\subparagraph}
{0pt}{3.25ex plus 1ex minus .2ex}{1.5ex plus .2ex}

\DeclarePairedDelimiter\abs{\lvert}{\rvert}%
\DeclarePairedDelimiter\norm{\lVert}{\rVert}%

\makeatletter
\let\oldabs\abs
\def\abs{\@ifstar{\oldabs}{\oldabs*}}
\let\oldnorm\norm
\def\norm{\@ifstar{\oldnorm}{\oldnorm*}}
\makeatother

\newtheorem{theorem}{Theorem}

\newtheorem{lemma}{Lemma}
\newtheorem{assumption}{Assumption}

\newcolumntype{L}[1]{>{\raggedright\let\newline\\\arraybackslash\hspace{0pt}}m{#1}}
\newcolumntype{C}[1]{>{\centering\let\newline\\\arraybackslash\hspace{0pt}}m{#1}}
\newcolumntype{R}[1]{>{\raggedleft\let\newline\\\arraybackslash\hspace{0pt}}m{#1}}
\newcolumntype{L}[1]{>{\raggedright\let\newline\\\arraybackslash\hspace{0pt}}m{#1}}
\newcolumntype{C}[1]{>{\centering\let\newline\\\arraybackslash\hspace{0pt}}m{#1}}
\newcolumntype{R}[1]{>{\raggedleft\let\newline\\\arraybackslash\hspace{0pt}}m{#1}} 

\begin{document}
\title{Sequential Search Models:  \\ A Pairwise Maximum Rank Approach\thanks{First version: April 28, 2021. I thank St\'{e}phane Bonhomme, Pradeep Chintagunta, Giovanni Compiani, Michael Dinerstein, Eyo Herstad, Ali Horta\c{c}su, Jonas Lieber, Dan Savelle, Myungkou Shin, Raluca Ursu and all participants in the consumer search digital seminar, the UChicago third year seminar, IO lunch, and metrics student group for many helpful comments. All errors are my own.}}
\author{ Jiarui Liu\thanks{University of Chicago. Email: jiarui@uchicago.edu }}

\date{this version: \today \\ \vspace{5mm}
PRELIMINARY AND INCOMPLETE}
\maketitle

\abstract{This paper studies sequential search models that (1) incorporate unobserved product quality, which can be correlated with endogenous observable characteristics (such as price) and endogenous search cost variables (such as product rankings in online search intermediaries); and (2) do not require researchers to know the true distribution of the match value between consumers and products. A likelihood approach to estimate such models gives biased results. Therefore, I propose a new estimator ---  pairwise maximum rank (PMR) estimator ---  for both preference and search cost parameters. I show that the PMR estimator is consistent using only data on consumers' search order among one pair of products rather than data on consumers' full consideration set or final purchase. Additionally, we can use the PMR estimator to test for the true match value distribution in the data. In the empirical application, I apply the PMR estimator to quantify the effect of rankings in Expedia hotel search using two samples of the data set, to which consumers are randomly assigned. I find the position effect to be \$0.11-\$0.36, and the effect estimated using the sample with randomly generated rankings is close to the effect estimated using the sample with endogenous rankings. Moreover, I find that the true match value distribution in the data is unlikely to be N(0,1). Likelihood estimation ignoring endogeneity gives an upward bias of at least \$1.17; misspecification of match value distribution as N(0,1) gives an upward bias of at least \$2.99.}

\normalsize
\doublespacing

\section{Introduction}

In many situations where people make choices, they might not have the information about all existing options for free. Acquiring and processing information can be costly in a wide range of empirical contexts: for example, school applications \citep{hoxbyturner}, housing and neighborhood \citep{bergmanchankapor2020}, hotels \citep{koulayev2014,gu2016,chenyao2017,ursu2018}, financial products \citep{hortacsu2004}, insurance products \citep{brown2002,honka2014}, credit market \citep{agarwal2020}, TV channels \citep{yao2017tv}, medical devices \citep{grennan2020}, automobiles \citep{gonzalezautomobile,yavorsky2021,murry2020}, other durable goods \citep{kim2010,kim2017,dongetal2020}, and so forth.

It is therefore important to take into account people's limited consideration and model their search process if we want to consistently estimate their preference and search costs in these scenarios. Consistent estimates of search costs allow us to quantify market frictions due to limited information. Moreover, consistent demand estimates are needed to analyze optimal firm behaviors, for example pricing and advertising decisions. Accurate demand and supply parameters are the basis for counterfactual social welfare analysis.

A large strand of literature models these empirical settings with sequential search and uses a likelihood approach to estimate preference and search cost parameters. These existing models assume away any unobserved product quality that could be correlated with observable characteristics (such as price) and search cost variables (such as product rankings in online search intermediaries), which causes endogeneity concerns. These models also assume that researchers know the true distribution of the match value between consumers and products. However, these assumptions can be violated in many empirical settings. 
 
This paper aims to address these issues by incorporating endogenous unobserved product quality into a canonical differentiated good sequential search model, as well as allowing researchers to be agnostic about the distribution of match value. A simulated likelihood estimation of such model gives biased estimates of both preference and search cost parameters. Therefore, I propose a new estimator --- pairwise maximum rank (PMR) estimator --- for consumers' preferences and search costs. I prove that this estimator is consistent under the more generalized model with endogeneity and mild assumptions on the match value distribution. I also show the performance of this estimator in monte carlo experiments.

The estimator is based on consumers' search behavior of an arbitrary pair of products. Intuitively, a consumer's search order among the product pair (i.e. which product out of the pair does this consumer search first) should on average be predicted by the consumer's search cost, as well as the observed and unobserved characteristics of the product pair. Consequently, consider any two consumers facing the same search cost and unobserved characteristics of the product pair. If their search orders among the product pair are different, then the difference should on average be explained by the observed characteristics of the product pair, which identifies the preference parameters of the observed characteristics. Similarly, comparing the search orders of any two consumers facing the same observed and unobserved product pair characteristics identifies the search cost parameters. 

The biggest advantage of the PMR estimator against a simulated likelihood estimator is that it is consistent when there is endogeneity concern. Endogeneity arises when unobserved quality of the product enters the consumer's utility of that product. The unobserved quality is likely to be correlated with observable characteristics of the product, such as price and advertising, because firms choose price and advertising potentially based on the product quality \citep{blp1995,shapiro2018,shapiro2020}. The unobserved product quality can also be correlated with consumers' search costs of that product. For example, higher quality products are likely to be positioned higher on a list page during online shopping, so consumers will have lower search costs for these products \citep{ursu2018}. Unobserved quality is often assumed away in existing sequential search models. Thus, likelihood estimations of these misspecified models without unobserved quality will lead to biased estimates of preference and search cost parameters, as I will show later in monte carlo simulations in Section \ref{sec:montecarlo}.   

A likelihood approach to estimate the sequential search model with unobserved quality is also problematic. The unobserved quality would have to be treated either as random effects or fixed effects. For random effects, we would have to assume a joint distribution of the unobserved quality and observed characteristics so that we can write out the likelihood function. These distributional assumptions are likely not well-founded by economic theory. Thus, most of current demand estimation literature treats the unobserved quality as fixed effects: it is allowed to be arbitrarily correlated with observed characteristics. But estimating such fixed effects with likelihood is extremely demanding in terms of computational resource and time when there are many products. 

The PMR estimator is able to treat the unobserved quality as fixed effects. The underlying assumption is that the unobserved quality of a product is the same for all consumers, which is consistent with most demand estimation literature. Thus, when we rank the search orders across consumers fixing a pair of products, the unobserved quality of the product pair is "differenced-out" because it is the same for all consumers. 

A second advanage of the PMR estimator is that it is robust to researchers' misspecification of the match value distribution. Current sequential search models have to assume that researchers know the true distribution of the match value error term in utility to make the likelihood estimation possible. \cite{yavorsky2021} show that misspecifying the distribution of match value leads to biased estimates of search costs, thus incorrect predictions of search behavior and welfare implication. I will corroborate their argument with monte carlo results in Section \ref{sec:montecarlo}. The PMR estimator is consistent as long as the match value distribution belongs to a large class of distributions, including all distributions assumed in existing literature so far, to the best of my knowledge. Moreover, a testable implication of the estimates can help reject false specifications of the match value distribution. This can also serve as a robustness check for assumptions of match value distributions in empirical settings.\footnote{More details are discussed in Section \ref{sec:nonparam}.}  

A third advantage of the PMR estimator is that it requires less data than the simulated likelihood estimator. The likelihood estimation often requires data on the full consideration set and the final purchase of each consumer for estimation and identification purpose. Yet the final purchase can be difficult to observe in online shopping contexts even if the full consideration set (clicking activity) is observed \citep{koulayev2014,brynjolfsson2010}. The full consideration set can be difficult to observe in most offline contexts, which prohibits the possibility of modeling consumer behavior as sequential search in many interesting economic contexts such as school and neighborhood choice. The PMR estimator, however, only requires data on the search order of an arbitrary pair of products, instead of final purchase or full consideration set data. This is because the PMR estimator only exploits variation across consumers within a product pair. 

A fourth advantage of the PMR estimator is that it is computationally lighter than the simulated likelihood estimator when there are many products. In order to compute the simulated likelihood, one would have to simulate and compute the reservation utilities of all products in all samples. When the number of products is very large, this will be computationally expensive. Thus, in practice researchers might choose to limit the total number of options available in the market for computational ease. Another reason for limiting the number of products is that the researchers might be unable to observe the characteristics or outcomes of all available products. Yet in contexts such as online shopping, there are usually a large number of options available, and it could be impractical to obtain data on all available options. Failing to account for all products might lead to biased estimates. This would not be an issue with the PMR estimator because it only requires data on the search order of one pair of products. Having data on more product pairs will increase the efficiency of the estimator, but it is not necessary for estimation or identification.  

For the empirical application, I apply a match value sequential search model with unobserved product quality to the Expedia hotel search data set without imposing a specific match value distribution assumption. The Expedia data set composes of two samples: a random ranking sample, where the positions of hotels on the list page are randomly assigned; and a non-random ranking sample, where products are ranked by relevance according to Expedia ranking algorithm, so the endogenous position variable can be correlated with unobserved hotel quality. In both samples, price can also be endogenous. I find that the PMR estimate of the position effects using the non-random ranking sample is close to the one using the random ranking sample, whereas the likelihood estimates are not as close. In fact, the likelihood estimates using the non-random ranking sample give an upward bias of at least \$1.17 in dollar values of the position effect due to position endogeneity. This is supporting evidence that the PMR estimator is able to recover the coefficients on endogenous variables better than the simulated likelihood estimator. Moreover, I find that the prevailing assumption of match value distribution as N(0,1) for likelihood estimation is not consistent with the data. Thus, the likelihood estimates using the random ranking sample give an upward bias of at least \$2.99 in dollar values of the position effect due to misspecification of the match value distribution.  

The rest of the paper proceeds as the follows. I first review related literature in Section \ref{sec:litreview}. Then I describe the model in Section \ref{sec:model} and present the estimator in Section \ref{sec:estimator}. Section \ref{sec:montecarlo} shows the monte carlo simulation results. Section \ref{sec:application} describes the empirical setting and discusses the results using the Expedia hotel search data. Lastly, I discuss some extensions of the baseline model and estimator in Section \ref{sec:extension}.

\section{Relevant Literature} \label{sec:litreview}

There has been a large body of literature on estimating sequential search models for differentiated goods in both economics and marketing using a likelihood approach \citep{kim2010,kim2017,koulayev2014,gu2016,chenyao2017,yao2017tv,de2017,morozov2020}. This paper contributes to this strand of literature by proposing a new estimator that can consistently estimate consumers' preference and search costs when there is endogenous unobserved product quality or when there is risk of researchers' misspecification of the match value distribution. \cite{ursu2018} deals with the endogeneity of product's position by estimating a sample where product's positions are randomly assigned. \cite{dongetal2020} uses panel data to estimate product intercepts with a likelihood approach, but there is limited discussion on the endogeneity of price or position. \cite{chungetal2019} proposes a likelihood based estimator that directly simulates the likelihood while forcing the random draws to satisfy the search set conditions. They point out that if we can reliably estimate the endogenous unobserved product quality as a first step, then their estimator can recover the true consumer preference and search cost parameters. \cite{yavorsky2021} identifies the variance of match value with exogenous cost shifters if the match value follows a normal distribution, whereas this paper considers a larger class of distributions for the match value.

Another strand of work uses aggregate level data instead of individual level consideration and purchase data to estimate sequential search models. \cite{gonzalezautomobile} use aggregate share combined with micro survey data and valid price instruments to perform a GMM estimation in the spirit of \cite{blp1995}. They assume particular distributions of the match value and the search cost for an analytical form of the search probabilities for computation. \cite{abaluck2020} use aggregate share and micro data, but they consider sequential search models where consumers search for observed (to researchers) characteristics, such as price, rather than unobserved (to researcher) characteristics, such as match value.\footnote{In their paper, consumers can be allowed to reveal their match value along with observed characteristics when they conduct a search, but they are not allowed to search only to reveal match value.} Their framework can incorporate price endogeneity if valid instruments are available. Their approach can also be used to estimate simultaneous search models.

Other works focus on estimating preference and search costs using different search models. \cite{honka2017} and \cite{de2012testing} estimate price search models and propose methods to empirically identify whether consumers are searching sequentially or simultaneously. \cite{honka2014} and \cite{bergmanchankapor2020} estimate simultaneous search models in the contexts of auto insurance and housing choice. \cite{hong2006} estimate the search cost distribution from a search model with homogeneous goods. \cite{hortacsu2004} considers search models with vertically differentiated goods.  

This paper is closely related to the maximum score literature. The PMR estimator builds on the seminal work on maximum score estimator by \cite{manski1975,manski1985,manski1987} and maximum rank correlation estimator by \cite{han1987}. \cite{horowitz1992} proposes a smoothed maximum score estimator. \cite{abrevaya2000} proposes a fixed effects maximum rank correlation estimator in panel data. The PMR estimator in this paper incorporates a pairwise fixed effect which represents the endogenous unobserved quality of the product pair. Moreover, the works aforementioned are in reduced form settings, whereas this paper proposes an estimator for a structural search model. Variants of the maximum score estimator have been applied to multinomial choice models \citep{fox2007,yan2013,pakesporter,shishumsong,khantamer}. This paper, however, studies a sequential search environment.

\section{Model} \label{sec:model}

We consider a model where consumers engage in sequential search for match values. Section \ref{sec:util} introduces the preference and search process of the consumer. Section \ref{sec:optimsearch} describes the optimal strategy of the consumer. Section \ref{sec:empiric} lays out the empirical specification of the model. Section \ref{sec:outcome} shows how to construct the observed outcome of the model to be used for estimation.

\subsection{Model Primitives} \label{sec:util}

Consider a set of consumers, denoted as $\mathcal{A}$, making discrete choices among a set of products $\mathcal{J} = \{0, 1, ..., J-1 \}$, where product $j=0$ is the outside option. Let the indirect utility $u_{aj}$ of consumer $a$ choosing product $j$ be composed of two parts: prior utility $\delta_{aj}$ and match value $\epsilon_{aj}$. The indirect utility is 
\begin{align} \label{eq:u}
u_{aj} = \delta_{aj} + \epsilon_{aj}  
\end{align}

The prior utility represents the consumer's valuation of the product prior to search. Before conducting any search, the consumer does not know his match value of any product, but he knows the distribution of the match value. To reveal the match value of any product $j$, he has to pay a search cost $c_{aj}$. After revealing the match value of a product, he decides whether to continue or stop searching. If he stops searching, he then makes a purchase decision with free recall. If he chooses not to purchase any inside good, then he gets a known utility of the outside good $u_{a0}$.  

\subsection{The Optimal Strategy} \label{sec:optimsearch}

To characterize the optimal strategy of consumers in the search context described above, I follow \cite{weitzman1979}. First, I define consumer $a$'s reservation utility, denoted as $r_{aj}$, of any inside good $j$ as the solution to the following equation: 
\begin{align} \label{eq:c}
c_{aj} = \int_{r_{aj}}^{\infty} \left( u - r_{aj} \right) dF_{aj}^u \left(u \right)  
\end{align}
where $F_{aj}^u$ is the distribution of utility $u_{aj}$. Intuitively, the reservation utility of a product is the utility level at hand that makes the consumer indifferent between searching that product or not. Assuming that the solution exists, we can then derive the optimal strategy as consisting of the following three steps:
\begin{enumerate}
\item Selection rule: after computing the reservation utility of all products, the consumer will search the product with the highest reservation utility among the unsearched products, if a search is to be made.
\item Stopping rule: stop searching if the highest realized utility so far exceeds the highest reservation utility of all unsearched products.
\item Choice rule: once stopped searching, the consumer chooses the product with the highest realized utility among all searched products.
\end{enumerate}
Note that the consumer is able to compute the reservation utility $r_{aj}$ from equation \eqref{eq:c} because he knows his search cost and the distribution of his utility $u_{aj}$. The distribution of utility is known prior to search because the consumer knows his prior utility $\delta_{aj}$ and the distribution of the match value $\epsilon_{aj}$. 

However, the distribution of the match value can be unknown to the researchers. The distribution can not be identified with the likelihood estimation, so it is often assumed to follow a particular form. Most literature assumes a normal distribution in order to obtain a simple analytical solution of the reservation utility for the ease of computation \citep{kim2010,chenyao2017,honka2017,ursu2018,chungetal2019,dongetal2020}.\footnote{\cite{gonzalezautomobile} and \cite{elberg2019} derive analytical solutions of the reservation utility for type I extreme value and logistic distributions of match value. \cite{yavorsky2021} show that the variance of match value can not be separately identified from the search cost without an exogenous search cost shifter if the match value follows a normal distribution. They have yet shown how to identify match value distribution when it is not normal.} Therefore, the distribution of the match value is very likely to be misspecified in a likelihood estimation. With the misspecified distribution, we will not be able to compute the correct reservation utility from equation \eqref{eq:c}. This will lead to biased estimates of the consumer's preference and search cost parameters in a simulated likelihood estimation.
 
To remediate the risk of misspecification of the match value distribution, we allow for a class of distributions that satisfy the following assumption:
\begin{assumption}\label{as:epsilon}
Let $r_{aj} - \delta_{aj}$ be on a finite support $ [ \underline{w}, \overline{w} ] \subset  \mathbb{R}$. The match value $\epsilon_{aj}$ is distributed iid across consumers as $F_{j}$ for any inside good $j$. Let the density $f_{j}$ be strictly positive on $[ \underline{w}, \overline{w} ]$. 
\end{assumption}
There are several things to note about this assumption. First, we allow the distribution of match value to be product-specific. This means that we can take into account potential correlation between the prior utility and the match value distribution of a product. For example, products with higher prior utility might have higher match value on average, or vice versa. Second, we are not assuming that the match value has bounded support; instead, we only assume that on the bounded support of $r_{aj} - \delta_{aj}$, the match value has strictly positive density. To the best of my knowledge, all distributions of match value assumed in the literature satisfy the assumption above: normal, type I extreme value, and logistic distribution with finite location and scale parameters. 

We can then rewrite equation \eqref{eq:c} as
\begin{align} \label{eq:cG}
c_{aj} =  \int_{r_{aj} - \delta_{aj}}^{\infty} \left[ \epsilon - \left( r_{aj} - \delta_{aj} \right) \right] f_{j} (\epsilon) d \epsilon
\end{align}
The right hand side of the equation above is just the marginal benefit of searching product $j$. To simplify notation, let us define the marginal benefit function of searching product $j$ as $G_j$. Formally,
\begin{align} \label{eq:G}
G_j(w) = \int_{w}^{\infty} (\epsilon-w) f_{j}(\epsilon) d \epsilon
\end{align}
Then equation \eqref{eq:cG} becomes
\begin{align}
c_{aj}  = G_j(r_{aj} - \delta_{aj})
\end{align}	
\begin{lemma} \label{lem:inv}
Let Assumption \ref{as:epsilon} hold, then $G_j$ is strictly decreasing on $[ \underline{w}, \overline{w} ]$. Moreover, $G_{j}^{-1}$ exists and is also strictly decreasing on its support.  
\end{lemma}

\begin{proof}
See appendix.
\end{proof}

Even if we do not know the exact functional form of $G_j^{-1}$, we know that it is strictly decreasing from Lemma \ref{lem:inv}. Thus, we can write the reservation utility as a function of the search cost and the prior utility:
\begin{align} \label{eq:rG}
 r_{aj} = G_j^{-1} ( c_{aj}) + \delta_{aj}
\end{align}	
We next discuss in more detail how the prior utility and the search cost are specified.

\subsection{Empirical Specification} \label{sec:empiric}

\subsubsection{Prior Utility} \label{sec:empiricutil}

I model the consumer's prior utility $\delta_{aj}$ as
\begin{align} \label{eq:delta}
\delta_{aj} =  x_{aj}' \beta + \xi_j + \nu_a + \eta_{aj}  
\end{align}
$x_{aj} \in \mathbb{R}^{q_x}$ are observable characteristics of product $j$ to consumer $a$. I treat a search impression as equivalent to a consumer, so the variation of observable characteristics (such as price) across consumers can come from the variation across search impressions. The baseline model only considers characteristics that vary across search impressions. The model with characteristics that are fixed across search impressions will be discussed in Section \ref{sec:invar}. The consumer's preference $\beta \in \mathbb{R}^{q_x}$ is a parameter of interest. We will allow $\beta$ to be consumer-specific in Section \ref{sec:hetero}.  

In addition to the observed characteristics, the prior utility also includes the unobserved (to researchers) product quality $\xi_j$. The unobserved quality can be correlated with observed characteristics. For example, firms will set their optimal prices taking into account the quality of the products. This is the classic endogenous price problem in the demand estimation literature. Here we do not impose any distributional assumption on the unobserved quality: unobserved quality can be arbitrarily correlated with observed characteristics. 

The remaining components of the prior utility are unobserved (to researchers) individual fixed effect $\nu_a$ and taste shock $\eta_{aj}$ prior to search. Pre-search taste shock could be a random recommendation from a friend, which is unobserved to the researchers but changes the consumer's valuation of the product before searching.\footnote{\cite{dongetal2020} also models pre-search taste shock in a similar fashion, but assumes it to follow a normal distribution.} Let the taste shock $\eta_{aj}$ be distributed iid across consumers and products with positive density everywhere on $\mathbb{R}$. Because the individual fixed effect $\nu_a$ does not affect the search order of any product pair within consumer, we don't impose any distributional assumption on $\nu_a$. 

With the prior utility specified as in equation \eqref{eq:delta}, consumer $a$'s indirect utility of choosing product $j$ is therefore
\begin{align}
u_{aj} = x_{aj}' \beta + \xi_j + \nu_a + \eta_{aj} + \epsilon_{aj}  
\end{align}

\subsubsection{Search Cost} \label{sec:empiricsrch}
The search cost incurred by the consumer involves the time and effort he spends to reveal the match value of the product of interest. Therefore, the search cost would depend on some search-related characteristics of the product as well as consumer characteristics. Specifically, I model the search cost as
\begin{align} \label{eq:cparam}
c_{aj} = \exp \left( z_{aj}' \gamma \right) 
\end{align}
where $z_{aj} \in \mathbb{R}^{q_z}$ are observable search cost variables. In the context of online search, $z_{aj}$ can include the position of the product on the list page. \cite{ursu2018} shows that products positioned further down on the list page get less clicks even when the positions are randomly assigned, suggesting "scrolling costs" of searching products further down on the page. Other examples of $z_{aj}$ are loading time and information complexity of the web page because consumers incur higher cognitive costs to gather and analyze information when the information is less ordered \citep{gu2016}. In the context of offline search, examples of $z_{aj}$ include consumers' distances to stores \citep{gonzalezautomobile,yavorsky2021}; how time-constrained consumers are when they search \citep{mcdevitt2014,pinnaseiler,chenyao2017}; other consumer characteristics that might affect search cost such as age, education, experience on searching: \cite{hortacsu2004} suggests that new investors in mutual funds who are younger, less educated and less experienced are likely to have higher information-gathering costs. The search cost parameter $\gamma \in \mathbb{R}^{q_z}$ is homogeneous across consumers in the baseline model. We will relax it to be consumer-specific in Section \ref{sec:hetero}. 

In the main specification of the model, I treat $z_{aj}$ as excluded from utility $u_{aj}$ for purpose of exposition. However, there can be some $z_{aj}$ that are included in utility $u_{aj}$. For example, distance to schools might affect both search costs and utility of students in a school choice context. Section \ref{sec:zutil} describes how we can adapt the model and the estimator to allow $z_{aj}$ to enter utility.

\subsection{The Observed Outcome} \label{sec:outcome}

This section discusses how I construct the outcome variable observed in the data for estimation. First, note that the reservation utility $r_{aj}$ can be derived from combining equations \eqref{eq:rG}, \eqref{eq:delta} and \eqref{eq:cparam}:
\begin{align} \label{eq:r}
 r_{aj} = G_j^{-1} (\exp \left( z_{aj}' \gamma \right) ) + x_{aj}' \beta + \xi_j +\nu_a + \eta_{aj}
\end{align}	
Next, define an outcome variable $S_{aij}$ as the following:
\begin{align} \label{eq:Sdef}
S_{aij} =  \mathds{1}\{ r_{ai}  > r_{aj}  \} \quad \forall a \in \mathcal{A}, \, i \neq j \in \mathcal{J}
\end{align}
This means that $S_{aij}=1$ if consumer $a$'s reservation utility of product $i$ is higher than that of product $j$, and $S_{aij}=0$ otherwise. Recall that the selection rule of the optimal strategy described in Section \ref{sec:optimsearch} implies that the consumer would search products in the decreasing order of their reservation utilities before he stops searching. Therefore, if consumer $a$ searches both $i$ and $j$, then the researcher observes $S_{aij}=1$ if product $i$ is searched before $j$ and $S_{aij}=0$ otherwise; if consumer $a$ searches $i$ but not $j$, then the researcher observes $S_{aij}=1$; if consumer $a$ searches $j$ but not $i$, then the researcher observes $S_{aij}=0$; if consumer $a$ does not search any of products $i$ and $j$, then the researcher can not observe $S_{aij}$. Let $\mathcal{A}_{ij}$ denote the set of consumers who search \textit{at least} one product among the product pair $i,j$. Then the researcher is able to observe the outcome variable $S_{aij}$ for all the consumers $ a \in \mathcal{A}_{ij}$. 

Combining equations \eqref{eq:r} and \eqref{eq:Sdef}, the outcome variable can be written more explicitly as
\begin{align} \label{eq:S}
S_{aij} =  \mathds{1}\{G_i^{-1} (\exp ( z_{ai}' \gamma ) ) - G_j^{-1} (\exp ( z_{aj}' \gamma ) ) + x_{ai}' \beta - x_{aj}' \beta + \xi_{i} - \xi_j + \eta_{ai} - \eta_{aj}  >0 \}
\end{align}
The consumer knows all elements in the equation above, but the researcher only observes $\left( S_{aij}, z_{ai}, z_{aj}, x_{ai}, x_{aj} \right)$ for all $a \in \mathcal{A}_{ij}$ and wants to estimate search cost parameter $\gamma$ and preference parameter $\beta$. For simplicity, I refer to the outcome variable $S_{aij}$ as the search order of product pair $i,j$ throughout the paper. To consistently estimate $\gamma$ and $\beta$, the researcher only needs data on the search order of any product pair, instead of data on consumers' full consideration set or final purchase. Note that if the researcher has data on consumers'  consideration set, then she can use the data by picking a searched product and an unsearched product to be the product pair of interest. The outcome variable (search order) of this product pair is observed because the searched product has higher reservation utility than the unsearched product.

\section{Estimator} \label{sec:estimator}

Given the estimating equation \eqref{eq:S}, I now propose a novel estimator for consumers' preference and search cost parameters. Section \ref{sec:nonsmooth} introduces the ideal estimator and shows strong consistency. Section \ref{sec:smooth} describes an implementable smoothed version of the estimator and proves strong consistency. 
 
\subsection{The Ideal Estimator} \label{sec:nonsmooth}

The estimator is developed in the spirit of the maximum score estimator \`a la \cite{manski1975} and the maximum rank correlation estimator \`a la \cite{han1987}. For any pair of products $i$ and $j$, the ideal estimator maximizes the following objective:
\begin{align*}
Q_{ \mathcal{A}_{ij}} (b,m) = {\abs{\mathcal{A}_{ij}} \choose 2}^{-1} \sum_{a \neq \tilde{a} \in \mathcal{A}_{ij}} \quad \quad \quad \quad \quad \quad \quad \quad \quad \quad \quad & \\
 \mathds{1}\{z_{ai}= z_{\tilde{a}i},z_{aj}= z_{\tilde{a}j} \}  \bigg[ \, \mathds{1}\{x_{aij}'b > x_{\tilde{a}ij}'b\} \mathds{1}\{S_{aij} > S_{\tilde{a}ij}\} 
																					    &+  \mathds{1}\{x_{aij}'b < x_{\tilde{a}ij}'b\} \mathds{1}\{S_{aij} < S_{\tilde{a}ij}\} \, \bigg] \\
+ \, \mathds{1}\{z_{ai}= z_{\tilde{a}i},x_{aij}= x_{\tilde{a}ij} \} \bigg[ \mathds{1}\{z_{aj}'m > z_{\tilde{a}j}'m  \} \mathds{1}\{S_{aij} > S_{\tilde{a}ij}\} 
																						&+ \mathds{1}\{ z_{aj}'m < z_{\tilde{a}j}'m \} \mathds{1}\{S_{aij} < S_{\tilde{a}ij}\} \bigg] \\
+ \, \mathds{1}\{z_{aj}= z_{\tilde{a}j},x_{aij}= x_{\tilde{a}ij} \} \bigg[ \, \mathds{1}\{ z_{ai}'m < z_{\tilde{a}i}' m \} \mathds{1}\{S_{aij} > S_{\tilde{a}ij}\} 
																						&+ \mathds{1}\{ z_{ai}' m > z_{\tilde{a}i}' m\} \mathds{1}\{S_{aij} < S_{\tilde{a}ij}\} \, \bigg] 
\end{align*}
where $x_{aij}$ denotes $x_{ai} - x_{aj}$ for any consumer $a \in \mathcal{A}_{ij}$. Let $\abs{\mathcal{A}_{ij}}$ be the cardinality of the set $\mathcal{A}_{ij}$. So ${\abs{\mathcal{A}_{ij}} \choose 2}$ is the number of ways to choose distinct consumer pairs.

The objective represents the frequency of correct predictions of the ranking of outcome variable $S$ within all consumer pairs, given the observed characteristics $x$ and search cost variables $z$. More specifically, to understand the first line of the objective: consider two consumers with the same search cost variables $z$ for both products, then the consumer with higher values of $x_{aij}' \beta$ should have higher values of $S_{aij}$ than the other consumer in expectation, according to equation \eqref{eq:S}. The reason is: first, the within product pair difference in unobserved quality $\xi_i - \xi_j$ is the same for both consumers; second, the random taste shock $\eta$ is distributed identically and independently across consumers and products. Intuitively, the consumer with larger $x_{aij}' \beta$ has larger pairwise difference in the prior utility than the other consumer in expectation. So if they have the same search costs for both products, then he has larger pairwise difference in the reservation utility, which implies that he is more likely to search product $i$ before $j$ than the other consumer. 

A similar understanding goes with the second line of the objective. If two consumers have the same search cost for product $i$ and same pairwise difference in observed characteristics, then the consumer with higher values of $z_{aj}' \gamma$ has higher values of $S_{aij}$ in expectation. This is because, according to equation \eqref{eq:S}, $S_{aij}$ is an increasing function of $z_{aj}' \gamma$. As we have shown in Section \ref{sec:optimsearch}, even if we don't know the exact functional form of $G_{j}^{-1}$, we have proved that it is strictly decreasing. Strict monotonicity is sufficient for identification because we only rely on the ranking within any consumer pair. The third line is the mirror argument of the second line.

Therefore, the true parameters $\beta$ and $\gamma$ make correct predictions of the ranking of outcome variable within consumer pairs most frequently, thus maximizing the objective function. If we observe the consideration set of consumers or the search order of multiple pairs of products, then we can use these additional data by summing the objective function over all pairs of products for a new objective. This objective will also be maximized at the true parameters. 

The maximizer of the objective above is the ideal estimator because if covariates $x$ and $z$ are continuously distributed, then the exact matching of covariates between two consumers almost never happens, so $\mathds{1}\{z_{ai}= z_{\tilde{a}i},x_{aij}= x_{\tilde{a}ij} \}$ is almost always zero. Therefore, we will propose an implementable smoothed version of the ideal estimator in Section \ref{sec:smooth} to deal with this concern.

Now, we show that the ideal estimator is strongly consistent and describe formally some additional sufficient assumptions. To simplify notation, we use  superscript $k$ to denote the $k$th component of a variable, and tilde on variable to denote all other components except for the $k$th one. 
\begin{assumption} \label{as:support} 
For a pair of product $(i,j)$ and any $a \neq \tilde{a} \in \mathcal{A}_{ij}$, \\
 (a) Supports of $x_{aij}-x_{\tilde{a}ij}$ and $z_{ai} - z_{\tilde{a}i}$ are not contained in any proper linear subspace of $\mathbb{R}^{q_x}$ and $\mathbb{R}^{q_z}$, respectively. \\
 (b) There exists at least one component $k$ such that $\gamma^k \neq 0$, $\beta^k \neq 0$, and
distributions of both $x_{aij}^k - x_{\tilde{a}ij}^k \, \vert \, w_x = ( \tilde{x}_{aij}, \tilde{x}_{\tilde{a}ij}, z_{ai}, z_{\tilde{a}i}, z_{aj}, z_{\tilde{a}j} )$ and
 $z_{ai}^k - z_{\tilde{a}i}^k \, \vert \, w_z = ( \tilde{z}_{ai}, \tilde{z}_{\tilde{a}i}, z_{aj}, z_{\tilde{a}j}, x_{aij}, x_{\tilde{a}ij} )$ have everywhere positive density on $\mathbb{R}$ for almost every values of $w_x$ and $w_z$. \\
 (c) There exists a known constant $\rho >0$ such that $\abs{\beta^k} / \norm{\beta} \geq \rho$ and $\abs{\gamma^k} / \norm{\gamma} \geq \rho$.\\ Let $k = 1$ without loss of generality.
\end{assumption}

Part (a) of Assumption \ref{as:support} resembles the full-rank condition in the maximum score literature. This means that we need variation of observed characteristics and search cost variables across consumers or search impressions. For the case where some observed characteristics are the same for all consumers, we will introduce an adapted estimator in Section \ref{sec:invar}. Most of the search cost variables would vary across consumers or search impressions, as discussed in Section \ref{sec:empiricsrch}. Part (b) is an assumption frequently made in the maximum score literature. One example of a component in $x$ that satisfies this assumption is price. A restriction implied by part (b) is that we must have at least one search cost variable with continuous support. Examples include log of distance to stores, log of time constraints, orderedness entropy as a measure of information complexity as in \cite{gu2016}, et cetera. 

The assumptions on the support and distribution of $z_{ai} - z_{\tilde{a}i}$ are made with respect to product $i$ without loss of generality. We could have made the assumptions with product $j$. The point here is that we only need these assumptions on search cost variables for one product within the pair. 

It is also important to realize that the ideal estimator is scale invariant. Therefore, we will only discuss identification with respect to the normalized parameters: $\beta^{*} = \beta / \norm{\beta}$ and $\gamma^{*} = \gamma / \norm{\gamma}$. The parameter space of interest is $\Theta_{\rho} = \{ (b,m) : \norm{b} = \norm{m} = 1, \abs{b^1} \geq \rho, \abs{m^1} \geq \rho \}$.\footnote{We can also normalize the parameters by setting one component at a fixed value, as in \cite{horowitz1992} and \\ \cite{abrevaya2000}. We use this normalization scheme in the monte carlo exercise in Section \ref{sec:montecarlo} and the empirical application in Section \ref{sec:application}.} 

\begin{theorem} \label{thm:nonsmooth}
Consider the model described in Section \ref{sec:model} and let Assumptions \ref{as:epsilon}-\ref{as:support} hold. Let $(b_{\mathcal{A}_{ij}}, m_{\mathcal{A}_{ij}})$ be the solution to $\max_{(b,m) \in \Theta_{\rho}} Q_{ \mathcal{A}_{ij}} (b,m)$, then $\lim_{\abs{\mathcal{A}_{ij}} \rightarrow \infty} b_{\mathcal{A}_{ij}} = \beta^{*}$ and \\ $\lim_{\abs{\mathcal{A}_{ij}} \rightarrow \infty} m_{\mathcal{A}_{ij}} = \gamma^{*}$ almost surely.
\end{theorem}

\begin{proof}
Proof of Theorem \ref{thm:nonsmooth} is in Appendix \ref{app:thm1pf}. 
\end{proof}

\subsection{The Smoothed Estimator} \label{sec:smooth}

As mentioned in Section \ref{sec:nonsmooth}, we now introduce an implementable smoothed version of the ideal estimator to deal with the lack of exact matching between continuously distributed covariates. The idea is to replace the indicator of exact matching with a smoothing function of the difference of the covariates to be matched. The smoothing function assigns more weights when the difference is closer to zero, and less weights when the difference is further away from zero. Formally, the smoothed estimator maximizes the following objective:
\begin{align*}
SQ_{ \mathcal{A}_{ij}} (b,m;  \sigma_{\mathcal{A}_{ij}}) = {\abs{\mathcal{A}_{ij}} \choose 2}^{-1} \sum_{a \neq \tilde{a} \in \mathcal{A}_{ij}} \quad \quad \quad \quad \quad \quad \quad \quad \quad  & \\
 K_{\sigma_{\mathcal{A}_{ij}}} \bigg( \tiny \norm{ \begin{bmatrix} z_{ai} \\ z_{aj} \end{bmatrix} - 
										 					   \begin{bmatrix} z_{\tilde{a}i} \\ z_{\tilde{a}j} \end{bmatrix} } \bigg) 
\normalsize
 \bigg[ \, \mathds{1}\{x_{aij}'b > x_{\tilde{a}ij}'b\} \mathds{1}\{S_{aij} > S_{\tilde{a}ij}\} 
	   &+  \mathds{1}\{x_{aij}'b < x_{\tilde{a}ij}'b\} \mathds{1}\{S_{aij} < S_{\tilde{a}ij}\} \, \bigg] \\
+ \,   K_{\sigma_{\mathcal{A}_{ij}}} \bigg( \tiny  \norm{ \begin{bmatrix} z_{ai} \\ x_{aij} \end{bmatrix} - 
										 \begin{bmatrix} z_{\tilde{a}i} \\ x_{\tilde{a}ij} \end{bmatrix} } \bigg) 
 \normalsize
\bigg[ \mathds{1}\{z_{aj}'m > z_{\tilde{a}j}'m  \} \mathds{1}\{S_{aij} > S_{\tilde{a}ij}\} 
   &+ \mathds{1}\{ z_{aj}'m < z_{\tilde{a}j}'m \} \mathds{1}\{S_{aij} < S_{\tilde{a}ij}\} \bigg] \\
+ \,  K_{\sigma_{\mathcal{A}_{ij}}} \bigg( \tiny \norm{ \begin{bmatrix} z_{aj} \\ x_{aij} \end{bmatrix} - 
										 \begin{bmatrix} z_{\tilde{a}j} \\ x_{\tilde{a}ij} \end{bmatrix} }  \bigg) 
\normalsize
 \bigg[ \, \mathds{1}\{ z_{ai}'m < z_{\tilde{a}i}' m \} \mathds{1}\{S_{aij} > S_{\tilde{a}ij}\} 
        &+ \mathds{1}\{ z_{ai}' m > z_{\tilde{a}i}' m\} \mathds{1}\{S_{aij} < S_{\tilde{a}ij}\} \, \bigg] 
\end{align*}
where $\norm{.}$ represents the Euclidean norm and $K_{\sigma_{\mathcal{A}_{ij}}} (v) = K (\nicefrac{v}{\sigma_{\mathcal{A}_{ij}}})$ for some function $K$ and scalar $\sigma_{\mathcal{A}_{ij}}$ that satisfy the following assumptions:
\begin{assumption} \label{as:K} 
(a) $K: \mathbb{R} \rightarrow \mathbb{R}$ is a continuous function such that $\abs{K(v)}<M$ for some finite $M$ and for all $v$ on $\mathbb{R}$. 
(b) $\lim_{v \rightarrow -\infty} K(v) = 0$ and  $\lim_{v \rightarrow \infty} K(v) = 0$.
\end{assumption}	
	
\begin{assumption} \label{as:bandwidth} 
$\{ \sigma_{\mathcal{A}_{ij}}\}_{\abs{{\mathcal{A}_{ij}}}=1}^{\,\infty}$ is a sequence of strictly positive scalars with $\lim_{\abs{\mathcal{A}_{ij}} \rightarrow \infty} \sigma_{\mathcal{A}_{ij}} = 0$.
\end{assumption}	

The smoothed estimator is closely related to the kernel density estimation methods and the smoothed maximium score estimator in \cite{horowitz1992}.\footnote{The optimal selection of bandwidth $\sigma_{\mathcal{A}_{ij}}$ can be derived analogously from \cite{horowitz1992}.} The following theorem proves the strong consistency of the smoothed estimator.
	
\begin{theorem} \label{thm:smooth}
Consider the model described in Section \ref{sec:model} and let Assumptions \ref{as:epsilon}-\ref{as:bandwidth} hold. Let $(b_{\mathcal{A}_{ij}}, m_{\mathcal{A}_{ij}})$ be the solution to $\max_{(b,m) \in \Theta_{\rho}} SQ_{ \mathcal{A}_{ij}} (b,m;  \sigma_{\mathcal{A}_{ij}})$, then $\lim_{\abs{\mathcal{A}_{ij}} \rightarrow \infty} b_{\mathcal{A}_{ij}} = \beta^{*}$ and $\lim_{\abs{\mathcal{A}_{ij}} \rightarrow \infty} m_{\mathcal{A}_{ij}} = \gamma^{*}$ almost surely.
\end{theorem}	

\begin{proof}
Proof of Theorem \ref{thm:smooth} is in Appendix \ref{app:thm2pf}.
\end{proof}

Note that the objective $SQ_{\mathcal{A}_{ij}}$ is still discontinuous with respect to the parameters of interest. If we want to deal with the complexity of discontinuity and increase the rate of convergence, we can apply a slightly adjusted smoothing function on the indicators involving the parameters of interest. For example, we can replace $\mathds{1}\{x_{aij}'b > x_{\tilde{a}ij}'b\}$ with $\tilde{K}_{\sigma_{\mathcal{A}_{ij}}}( x_{aij}'b - x_{\tilde{a}ij}'b)$ where $\lim_{v \rightarrow -\infty} \tilde{K}(v) = 0$ and $\lim_{v \rightarrow \infty} \tilde{K}(v) = 1$.
	
\section{Monte Carlo Simulations} \label{sec:montecarlo}

I use a monte carlo experiment to study the performance of the smoothed PMR estimator proposed in Section \ref{sec:smooth}, and compare it with that of the simulated likelihood estimator. I simulate 500 data sets from a sequential search model as described in Section \ref{sec:model} with 5000 consumers and 30 products. Product 1 is considered as the outside option. 

The unobserved quality $\xi$ of a product is correlated with both its price (in observed characteristics $x$) and its position (in search cost variable $z$). More specifically, the endogenous observed characteristics (price) is positively correlated with the unobserved quality; the endogenous search cost variable (position of product on the list page) is negatively correlated with unobserved quality\footnote{Products with higher quality $\xi$ are ranked higher, thus having smaller values of the position variable.}. I let price and position to be the only endogenous variables in $x$ and $z$ for simplicity. The true distribution of the match value $\epsilon$ is N(0,3). Consumers' taste shock $\eta$ is generated from N(0,0.5). The individual fixed effect $\nu$ is abstracted away in the simulated data set because it drops out when we construct the outcome variable by taking the difference of reservation utilities between the product pair of interest.  

The smoothing function $K$ of the PMR estimator is the pdf of standard normal. Bandwidth is set to be $N^{-\frac{1}{5}}$, where $N$ is the number of consumer pairs used to compute the estimator. The number of consumer pairs varies across simulations because the set of consumers who end up searching the product pair of interest varies across simulated data sets. On average, the number of consumer pairs across data sets is 1700. We don't have ${5000 \choose 2}$ number of consumer pairs because we only use consumers who search at least one of the product pair of interest. 

The simulated likelihood estimator utilizes consumers full consideration set data, thus having significantly more observations than the PMR estimator\footnote{Ideally, we want to compare results when both PMR and likelihood estimators use data on the search order of one product pair, but likelihood estimation with such data is not possible. Table \ref{tab:PMRLL} shows that PMR performs better than likelihood estimator even with less data. Note also that the PMR estimator can use full consideration set data by adding up the objectives of all product pairs and gain more efficiency.}. I follow the literature and compute the simulated likelihood using the logit-smoothed AR simulator with 50 draws of the match value from the assumed distribution and a scaling factor of $1/3$. I consider two scenarios of the simulated likelihood estimation: match values are drawn from (1) the misspecified distribution N(0,1), and (2) the true distribution N(0,3). Importantly, both scenarios ignore the presence of endogenous unobserved product quality.

Table \ref{tab:PMRLL} shows that the likelihood estimates in both scenarios have larger absolute bias than the PMR estimates for the coefficients on the endogenous price and the endogenous position variables. This is because both scenarios fail to take into account the endogenous unobserved product quality. The scenario with the true match value distribution (MVD) has smaller absolute bias than the one with the misspecified MVD. Thus, comparing likelihood estimates using the true MVD with those using the misspecified MVD helps us understand how misspecification of MVD biases the results; comparing the PMR estimates with the likelihood estimates using the true MVD helps us understand how ignoring endogeneity biases the results.\footnote{The likelihood estimates on $\gamma_{position}$ with true MVD has relatively small absolute bias because I generate the data such that the correlation between position and unobserved quality is small conditional on price.}

The likelihood estimates of the endogenous price coefficient is upward biased. Intuitively, the product with higher price is likely to have higher unobserved product quality, but price and product quality work in different directions on the reservation utility. So if we ignore the presence of unobserved product quality, then we would wrongly think that price has a smaller absolute effect on the reservation utility. This leads to an upward biased price coefficient because price coefficient is negative. Moreover, misspecifying the MVD from N(0,3) to N(0,1) leads to an upward biased estimate of the position coefficient. This is because the true $G^{-1}$ function is steeper than the misspecified $G^{-1}$ function for most search cost values generated in the simulated data set (see Figure \ref{fig:Ginv_simulation}). If a product is moved to a larger position, then the associated decrease in reservation utility due to the decrease in the true $G^{-1}$ function would be wrongly attributed to the position coefficient. Therefore, we would wrongly conclude the position effect to be larger than the truth.

\begin{table}
\begin{threeparttable}[H]
\caption{Simulation Results of Smoothed PMR vs Simulated Likelihood}
\label{tab:PMRLL}
\begin{tabular}{lcc|cccc}
\hline \hline
& \multicolumn{2}{c}{\text{PMR}} & \multicolumn{4}{c}{\text{Likelihood}}  \\
& & &  \multicolumn{2}{c}{\textit{Misspecified MVD}}  &  \multicolumn{2}{c}{\textit{True MVD}} \\
&\textbf{$\beta_{price}$} & \textbf{$\gamma_{position}$} &\textbf{$\beta_{price}$}&\textbf{$\gamma_{position}$} &\textbf{$\beta_{price}$}&\textbf{$\gamma_{position}$} \\ \hline
\text{True Value} & -1 & 0.2 & -1 & 0.2 & -1 & 0.2 \\
\text{Mean Bias}&0.0040&-0.0261&0.5311&0.1044&0.4118&-0.0612\\
\text{MSE}&0.0638&0.0104&0.3063&0.0112&0.1762&0.0038\\ 
\text{Mean Obs}&1,700 &1,700 & 150,000 & 150,000 & 150,000 & 150,000 \\ \hline \hline
\end{tabular}
\begin{tablenotes}[normal,flushleft]
\item \textit{Notes} This table shows estimates of coefficients of the endogenous variables: price and position using 500 simulated data sets. The smoothed PMR estimator uses data on the search order of one product pair, and the simulated likelihood estimator uses full consideration set data. Likelihood estimation assumes away endogeneity. True match value distribution (MVD) is N(0,3) and misspecified MVD is N(0,1).
\end{tablenotes}
\end{threeparttable}
\end{table}

\begin{figure}[hp]
\caption[]{$G^{-1}$ of True vs. Misspecified Match Value Distribution in Simulation}
\label{fig:Ginv_simulation}
\centering
\begin{minipage}{.8\linewidth}
  \includegraphics[width=\linewidth]{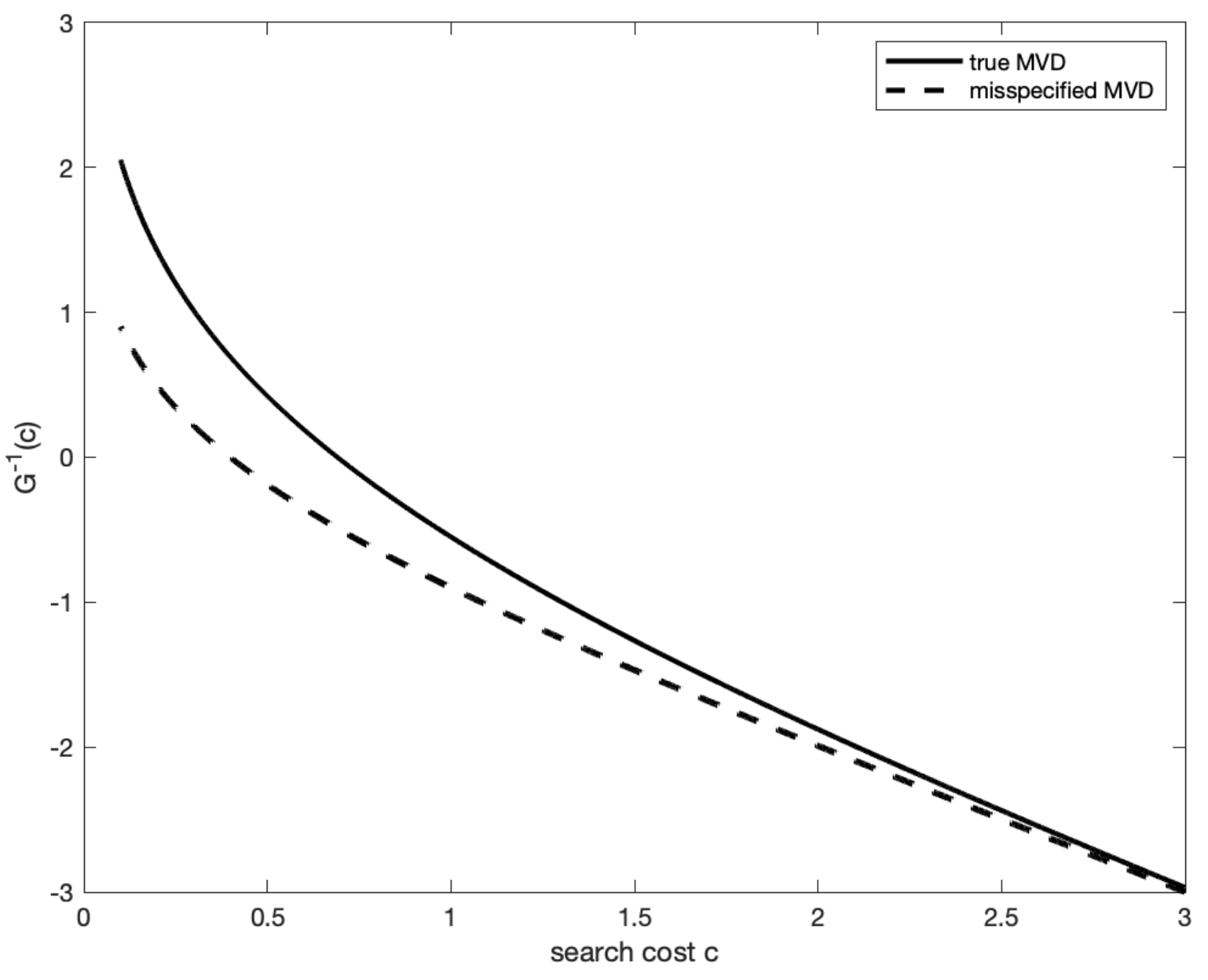} \\
\footnotesize
\emph{Notes} This figure plots the the $G^{-1}$ associated with the true match value distribution N(0,3) in the monte carlo simulation, against the $G^{-1}$ of the misspecified match value distribution N(0,1).
\end{minipage}
\end{figure} 

\section{Empirical Application} \label{sec:application}

This section lays out an empirical application using a data set on consumers searching for hotels online over an eight-month
period between November 1, 2012, and June 30, 2013. The data set is provided by a leading Online Travel Agency, Expedia\footnote{The dataset is available at \url{www.kaggle.com/c/expedia-personalized-sort/data}.}. The data set is briefly described in \ref{sec:kaggle_data}. Section \ref{sec:kaggle_est} shows the estimated preference and search cost parameters using both the PMR estimator and the simulated likelihood estimator, and discusses various sources of biases of the likelihood estimator. Section \ref{sec:kaggle_ginv} presents the estimated $G^{-1}$ function, and compare it with the implied $G^{-1}$ function by assuming a specific match value distribution. 

\subsection{Data} \label{sec:kaggle_data}

We now explain how consumers search for hotels on Expedia and what variables the data set contains. First, the consumer submits a search query on Expedia by specifying details of his trip, such as the destination (city, country), the travel dates, and the number of travelers and rooms requested. I observe all the variables above, in addition to booking window (the number of days before the beginning of the trip). Second, the consumer gets a search impression that contains an ordered list of hotels and their characteristics (list page). I observe the hotel ID, its position on the list page, and its characteristics (price, star rating, review score, location score, chain, and promotion indicator). Third, the consumer can click on a hotel to reveal his match value of that hotel. Then he can either return to the previous screen to click on another hotel, leave the site without purchasing, or make a purchase. I observe all clicks and purchases consumers make. The dataset contains 4.5 million observations of 166,036 search impressions for 788 destinations in total. For the empirical application, I use the 4 largest destinations following the literature of estimating sequential search models using the Expedia data set. More details of the data set can be found in \cite{ursu2018}. 

The data set randomly assigns consumers to two different ranking algorithms of products on the search impression page:  (1) a random ranking, where products are ranked randomly, and (2) the Expedia ranking, where products are ranked by relevance. Table \ref{tab:sumstatD4} shows summary statistics of the random ranking and the non-random ranking samples used in the estimation. The two samples are mostly balanced on hotel level characteristics. On the impression level characteristics, the random ranking sample has longer booking windows on average, which might cause the concern that consumers are not randomly assigned to different ranking algorithms. However, \cite{ursu2018} performs randomization checks in its Online Appendix B by showing that most of the consumer characteristics across the two samples are not significantly different, and for the characteristics that are significantly different, the magnitudes are very small. It also shows that the hotel positions are randomly generated in the random ranking sample. The summary statistics of all 788 destinations are in Table \ref{tab:sumstat788}, which are quite similar to the estimation sample in Table \ref{tab:sumstatD4}. 

\begin{table}
\begin{threeparttable}[H]
\caption{Summary Statistics, Random Ranking vs. Non-Random Ranking}
\label{tab:sumstatD4}
\begin{tabular}{lcccccc}
\hline \hline
\textit{Random ranking sample} &  \\
& Mean & SD & Median & Min & Max & Obs \\
\textit{Hotel level} &  \\
\quad \text{Price} (\$100) &1.77&1.19&1.49&0.17&10.00&121,813 \\
\quad \text{Star rating} &3.56&0.93&4.00&1.00&5.00&121,813 \\
\quad \text{Review score} &3.95&0.72&4.00&0.00&5.00&121,813 \\
\quad \text{Chain} &0.76&0.43&1.00&0.00&1.00&121,813 \\
\quad \text{Location score} &3.65&1.29&3.99&0.00&5.97&121,813 \\
\quad \text{Promotion}  &0.38&0.49&0.00&0.00&1.00&121,813 \\
\quad \text{Position} &18.16&10.71&18.00&1.00&39.00&121,813 \\
\textit{Impression level} &  \\
\quad \text{Number of hotels displayed} &28.09&8.13&32.00&5.00&36.00&4,336 \\
\quad \text{Booking window (days)} &61.51&66.25&37.00&0.00&487.00&4,336 \\
\quad \text{Click} &1.15&0.72&1.00&1.00&15.00&4,336  \\
\quad \text{Purchase} &0.06&0.24&0.00&0.00&1.00&4,336 \\ \hline 
\textit{Non-Random ranking sample} &  \\
& Mean & SD & Median & Min & Max & Obs \\
\textit{Hotel level} &  \\
\quad \text{Price} (\$100) &1.89&1.14&1.69&0.17&9.99&195,640 \\
\quad \text{Star rating} &3.74&0.82&4.00&1.00&5.00&195,640 \\
\quad \text{Review score}  &4.05&0.62&4.00&0.00&5.00&195,640 \\
\quad \text{Chain} &0.68&0.47&1.00&0.00&1.00&195,640 \\
\quad \text{Location score} &4.25&1.22&4.22&0.00&5.97&195,640 \\
\quad \text{Promotion} &0.48&0.50&0.00&0.00&1.00&195,640 \\
\quad \text{Position} &18.58&10.74&19.00&1.00&40.00&195,640 \\
\textit{Impression level} &  \\
\quad \text{Number of hotels displayed} &29.69&6.79&32.00&5.00&38.00&6,590 \\
\quad \text{Booking window (days)} &39.23&54.05&19.00&0.00&443.00&6,590  \\
\quad \text{Click} &1.13&0.69&1.00&1.00&20.00&6,590 \\
\quad \text{Purchase} &0.88&0.32&1.00&0.00&1.00&6,590 \\ \hline \hline
\end{tabular}
\begin{tablenotes}[normal,flushleft]
\item \textit{Notes} This table shows the summary statistics of the 4 largest destinations of the Expedia data set, for both the random ranking sample and the non-random ranking sample. 
\end{tablenotes}
\end{threeparttable}
\end{table} 

When constructing search impressions used for estimation\footnote{Since search impressions by the same consumer cannot be linked in the data set, we treat a search impression and a consumer as the same.}, I follow \cite{ursu2018} and \cite{chungetal2019} by removing impressions including any hotel with unrealistically high or low price per night (less than \$10 or more than \$1000 per night), and impressions including any hotel with potential price error (total price paid exceeding 130\% of price per night multiplied by the number of nights). Star ratings are assigned by Expedia according to the type of hotel, the level of luxury, and the amenities provided. Review score is the average of the review scores from consumers who made reservations for the hotel on Expedia in the past. Chain is an indicator for whether the hotel belongs to a Chain. Location score designed by Expedia ranges from 0 to 7, and measures the hotel's location centrality and surrounding amenities, et cetera. Promotion is an indicator variable for whether the hotel has an ongoing promotion.

\subsection{Preference and Search Cost Estimates} \label{sec:kaggle_est}

\subsubsection{PMR Estimates}

We are able to construct the outcome variable for smoothed PMR estimation because we observe the consideration set of the consumers, i.e. the hotels they click on. For any pair of hotels, we only keep consumers who click on exactly one of the hotel in the pair. For these consumers, we know that their reservation utility of the hotel being clicked on must be higher than the other hotel of the pair because the other hotel is not searched.  

The search cost variables include the position of the hotel on the list page and the booking window of the consumer. \cite{ursu2018} shows that position of the hotel affects the consumer's search cost and not the prior utility, using the random ranking sample of the same data set. I allow booking window to potentially affect search cost because when a consumer is closer to his check-in date, each time unit spent searching could be more costly. Similar arguments appear in \cite{chenyao2017}, \cite{mcdevitt2014} and \cite{pinnaseiler} in contexts of online hotel search, grocery shopping and home services. The length of booking window should not affect the consumer's preference for the hotel because his utility for the hotel only realizes when he checks in not when he makes the reservation.\footnote{See \cite{chenyao2017} for a similar argument.}There might be concerns that the booking window affects how the inside goods compare with the outside option. But since the PMR estimator is consistent fixing any product pair of our choice, we can choose product pairs that do not involve the outside option. More broadly, even in the case where the booking window affects the consumer's prior utility, we can still apply the adapted estimator in Section \ref{sec:zutil} for estimation. 

The observed characteristics of the hotels prior to search include price per night, promotion indicator, star rating, review score, location score and hotel chain indicator. Among the observed characteristics and search cost variables, price per night, promotion indicator, position and booking window vary across search impressions, whereas star rating, review score, chain indicator and location score are invariant across search impressions during the time frame when the data was collected. The preference for impression-invariant variables is estimated using the estimator described in Section \ref{sec:invar}.\footnote{In order to apply the estimator in Section \ref{sec:invar}, we assume that position and price are the potential endogenous variables and two hotels with the same price and position have iid distributed unobserved quality, so Assumption \ref{as:endvar} holds.} I use the estimator described in Section \ref{sec:nonparam} to estimate both the preference for the impression-varying variables and the $G^{-1}$ function implied by the match value distribution in the data. This version of the estimator only requires normalization of the preference coefficient, rather than both preference and search cost coefficients as described in Section \ref{sec:smooth}. Thus, I normalize the price coefficient to be -1 throughout the estimation.

To increase the efficiency of the estimator, I sum up the objectives of all possible hotel pairs from eligible hotels. Eligible hotels are hotels with more than 100 search impressions. This results in 626 hotel pairs in the random ranking sample and 1,558 hotel pairs in the non-random ranking sample. Thus, the number of observations I use to estimate the impression-varying variables is the sum of observations across all hotel pairs: 2,746,687 for the random ranking sample and 5,515,197 for the non-random ranking sample. This means that each hotel pair has on average 4,387 and 3,540 observations, respectively. The impression-invariant variables are estimated by summing up the objectives of all eligible search impressions. Eligible search impressions are impressions where more than 24 hotels are displayed on the list page. This results in 3,484 and 5,752 search impressions in the two samples. The number of observations used to estimate the impression-invariant variables is the sum of observations across all search impressions: 121,492 and 198,339 for the two samples. This implies that each search impression has on average 34 observations for both samples. 

\begin{table}
\begin{threeparttable}[H]
\caption{PMR Estimates}
\label{tab:kaggle_PMR}
\begin{tabular}{lcc}
\hline \hline
& \textit{Random Ranking Sample} & \textit{Non-Random Ranking Sample} \\ 
& & \\ 
\textit{Search Cost} & &  \\
\quad \text{Position}  &0.0012$^{***}$ & 0.0034$^{***}$  \\
& (0.0005, \hspace{0.3mm} 0.0017) & (0.0027, \hspace{0.3mm} 0.0840)  \\
\quad \text{Booking window (days)} &-0.6120 &-0.5240  \\
& (-0.6176, \hspace{0.3mm} 0.1772) & (-1.2438, \hspace{0.3mm} 0.4415)  \\
& & \\
\textit{Utility} & & \\
\quad \text{Price} (\$100) & -1 & -1  \\
& & \\
\quad \text{Star rating}  &1.0956$^{***}$  & 0.6133$^{***}$  \\
& (0.8744, \hspace{0.3mm} 1.3883) &  (0.4966, \hspace{0.3mm} 0.6530)  \\
\quad \text{Review score} &-0.1731$^{**}$  &0.0210  \\
& (-0.3350, \hspace{0.3mm} -0.0093) &  (-0.1014, \hspace{0.3mm} 0.0899)  \\
\quad \text{Chain} &0.1481$^{*}$ &0.1985$^{***}$  \\
& (-0.0043, \hspace{0.3mm} 0.4516) &  (0.1030, \hspace{0.3mm} 0.2878)  \\
\quad \text{Location score} &0.5358$^{***}$ &0.3304$^{***}$ \\
& (0.3071, \hspace{0.3mm} 0.5931) &  (0.2518, \hspace{0.3mm} 0.3740)  \\
\quad \text{Promotion} &0.1205$^{**}$   &0.0427  \\
& (0.0028, \hspace{0.3mm} 0.2073) &  (-0.0281, \hspace{0.3mm} 0.0780)  \\
& &  \\
\textit{Number of Observations} & &  \\
\quad \text{Impression-varying variables}  & 2,746,687  & 5,515,197 \\
\quad \text{Hotel pairs used}  & 626  & 1,558  \\
\quad \text{Impression-invariant variables} & 121,492  & 198,339  \\
\quad \text{Impressions used} & 3,484  & 5,752  \\
\quad \text{Total} & 2,868,179 & 5,713,536  \\ \hline \hline
\end{tabular}
\begin{tablenotes}[normal,flushleft]
\item \textit{Notes} This table shows the PMR estimates using the random ranking sample and the non-random ranking sample of the 4 largest destinations of the Expedia data. Impression-varying variables include: position, booking window, price, promotion. Impression-invariant variables include: star rating, review score, chain, location score. The parentheses are 95 percent confidence intervals computed using subsampling. Subsampling uses 350 replications. p-values are computed from two sided tests. $^* p < 0.1$; $^{**} p < 0.05$; $^{***} p < 0.01$.
\end{tablenotes}
\end{threeparttable}
\end{table} 

I apply the PMR estimator to both the random ranking sample and the non-random ranking sample. In the non-random ranking sample, the hotel's position can be correlated with the unobserved hotel quality due to Expedia's ranking algorithm, which creates one source of endogeneity. However, the position endogeneity does not arise in the random ranking sample because the hotel's position is randomly generated in this sample.\footnote{More details of the Expedia ranking algorithm and the randomized ranking algorithm can be found in \cite{ursu2018}.} Another source of endogeneity is that the hotel's price can also be correlated with the unobserved hotel quality. The price endogeneity is potentially present in both samples. 
 
Table \ref{tab:kaggle_PMR} shows the PMR estimates for each of the samples. The estimated search cost and preference parameters in the non-random ranking sample are mostly close to those in the random ranking sample.\footnote{The estimates across two samples using PMR are much closer than those using simulated likelihood if we compare Table \ref{tab:kaggle_PMR} and Table \ref{tab:kaggle_LL}.} This is consistent with the argument that the PMR estimator can correct for endogeneity bias. The results show that the consumers incur higher search costs if they search hotels further down on the list page. Their search costs are also higher if they are closer to the actual travel date, thus more time constrained, although the estimates are relatively noisy. The results also show that consumers prefer hotels with higher star rating, higher location score, belonging to a hotel chain or on promotion. Yet, the review scores have no statistically significant effect or slightly negative effective on consumers' utility. This is potentially due to the nonlinearity of the effects of review scores.\footnote{\cite{ursu2018} and \cite{chungetal2019} also find mostly negative or non statistically significant effects of review scores using each destination of the random ranking sample.} 

The inference is conducted using the subsampling procedure, following the literature on maximum score estimators \citep{fox2007,fox2013}. I construct the 95 percent confidence intervals as in \cite{politis1994}. The centered p-values for the two sided tests are computed based on \cite{berg2010}. The exception is the position estimates of the non-random ranking sample. Since the empirical distribution of the estimates from subsamples has a very long tail, I construct its confidence interval from the raw empirical distribution and compute the uncentered p-values.

\subsubsection{Simulated Likelihood Estimates}

\begin{table}
\begin{threeparttable}[H]
\caption{Simulated Likelihood Estimates}
\label{tab:kaggle_LL}
\begin{tabular}{lcccc}
\hline \hline
& \multicolumn{2}{c}{\textit{Random Ranking Sample}} & \multicolumn{2}{c}{\textit{Non-Random Ranking Sample}} \\ 
& \multicolumn{2}{c}{\textit{Estimates}} & \multicolumn{2}{c}{\textit{Estimates}} \\ 
& & & & \\ \hline
& Raw & Normalized & Raw & Normalized \\

\textit{Search Cost} & & & & \\
\quad \text{Position} & 0.0079$^{***}$ &  & 0.0291 &  \\
& (0.0001) & & (0.0527) & \\
\quad \text{Booking window (days)} & -0.0014$^{***}$ & & 0.0031&  \\
& (0.0000) & & (0.0975) & \\
& & & & \ \\
\textit{Utility} & & & &\\
\quad \text{Price} (\$100) & -0.2556$^{***}$ & -1 & -0.7113$^{***}$ & -1 \\
& (0.0182) & & (0.0426) & \\
\quad \text{Star rating} & 0.1822$^{***}$ & 0.7127 & 0.3874$^{***}$ & 0.5447 \\
& (0.0113) & & (0.0174) & \\
\quad \text{Review score} & -0.0843$^{***}$ & -0.3299 & 0.3422$^{***}$ & 0.4811 \\
& (0.0089) & & (0.0122) & \\
\quad \text{Chain} & -0.0391$^{**}$ & -0.1530 & 0.0821$^{***}$ & 0.1155 \\
& (0.0164) & & (0.0123) & \\
\quad \text{Location score} & 0.1034$^{***}$ & 0.4044 & 0.4017$^{***}$ & 0.5647 \\
& (0.0093) & & (0.0237) & \\
\quad \text{Promotion} & 0.1214$^{***}$ & 0.4748 & -0.0012 & -0.0017 \\
& (0.0195) & & (0.0140) & \\
& & & & \ \\
\textit{Number of Observations} & & & & \\
\quad \text{Impressions} & 4,336 & 4,336 & 6,590 & 6,590  \\
\quad \text{Total} & 121,813 & 121,813 & 195,640 & 195,640 \\ \hline \hline
\end{tabular}
\begin{tablenotes}[normal,flushleft]
\item \textit{Notes} This table shows the raw and normalized (by price coefficient) simulated likelihood estimates, using the random ranking sample and the non-random ranking sample of the 4 largest destinations of the Expedia data. Estimation assumes match value distribution is N(0,1). The parentheses are the standard errors. $^* p < 0.1$; $^{**} p < 0.05$; $^{***} p < 0.01$.
\end{tablenotes}
\end{threeparttable}
\end{table} 

Table \ref{tab:kaggle_LL} shows the simulated likelihood estimates for both the random ranking sample and the non-random ranking sample. The estimation assumes that the match value is distributed as N(0,1) and the pre-search taste shock is also distributed as N(0,1). I use a logit-smoothed AR simulator with 50 draws and a scaling factor of $1/3$ to simulate the probabilities, following \cite{ursu2018}. In order to compare the simulated likelihod estimates with the PMR estimates, I also show the utility parameters after normalizing the price coefficient to be -1. This also makes the estimates more interpretable in terms of the dollar values. 

Comparing results from the two samples, we find that both the raw search cost estimates and the normalized preference estimates are quite different across samples. More importantly, most of the differences are much larger than those of the PMR estimates in Table \ref{tab:kaggle_PMR}.\footnote{The exceptions are star rating and booking window. However, both PMR and likelihood estimates of booking window are noisy, thus direct comparison of point estimates is less informative.} This suggests that the simulated likelihood estimates are more susceptible to the bias caused by position endogeneity. More specifically, the raw estimate of price coefficient in the non-random ranking sample is downward biased compared to the estimate in the random ranking sample. The bias could arise from the fact that position, price and unobserved quality are correlated with each other in the non-random ranking sample. Moreover, the raw estimates of the preference parameters in the non-random ranking sample are upward biased compared to those in the random ranking sample, except promotion indicator. The upward bias comes from the fact that hotels with better characteristics (i.e. higher star rating, review and location score) are ranked higher on the list page under the Expedia's ranking algorithm. Since the higher ranked hotels are likely to have higher unobserved quality, this will increase the reservation utilities of these hotels even more. So if we ignore the presence of the unobserved quality, we would wrongly attribute its effect on reservation utility to the hotel characteristics, causing an upward bias.\footnote{Note that this does not contradict with Assumption \ref{as:endvar} because the other hotel characteristics are correlated with unobserved quality only through the endogenous price and position variables.} 

\subsubsection{The Position and Booking Window Effects: PMR versus Likelihood}

To understand the estimated position and booking window effects using both estimators, we need to look at their effects in dollar values rather than in raw estimates. Unlike the preference parameters, we can not simply divide the raw estimates of position and booking window by the estimated price coefficients because they affect search costs exponentially rather than linearly. Moreover, because the effects are exponential, one unit increase of position (booking window) at different slots (days) has different dollar values. 

Table \ref{tab:dollareffect} shows the position (booking window) effects in dollar values at slot (day) 1, 5, 10 and 20. It suggests that the position effect increases as consumers move further down on the list page. More interestingly, consumers who search for hotels one day before the trip have much higher search costs than those who search two days before, according to the PMR estimates. The booking window effect decreases drastically as consumers search for hotels more days in advance: if they search 20 versus 21 days before the check-in, their search costs are almost the same. 

\begin{table}
\begin{threeparttable}[H]
\caption{The Dollar Values of the Effects of Position and Booking Window}
\label{tab:dollareffect}
\begin{tabular}{lcccc}
\hline \hline
& \multicolumn{2}{r}{\textit{Random Ranking Sample}} & \multicolumn{2}{c}{\textit{Non-Random Ranking Sample}} \\ 
& & & & \\ \hline
& PMR  & Likelihood  & PMR  & Likelihood  \\
& & & & \\
\multicolumn{2}{l}{\textit{Position increases by 1 slot}} & & & \\
\quad \text{At slot 1}  &0.1153&3.1094&0.3415&4.2786 \\
\quad \text{At slot 5}&0.1158&3.2086&0.3461&4.8073\\
\quad \text{At slot 10}&0.1165&3.3372&0.3521&5.5611\\
\quad \text{At slot 20}&0.1178&3.6098&0.3642&7.4417\\
& & & & \\
 \multicolumn{2}{c}{\textit{Booking Window increases by 1 day}} & & &\\
\quad \text{At day 1} &-24.8215&-0.5550&-24.1508&0.4367 \\
\quad \text{At day 5}&-2.1460&-0.5519&-2.9693&0.4422\\
\quad \text{At day 10}&-0.1006&-0.5480&-0.2162&0.4490\\
\quad \text{At day 20}&-0.0002&-0.5402&-0.0011&0.4631\\
\hline \hline
\end{tabular}
\begin{tablenotes}[normal,flushleft]
\item \textit{Notes} This table shows the dollar value effects of position and booking window computed from the PMR and simulated likelihood estimates in Table \ref{tab:kaggle_PMR} and \ref{tab:kaggle_LL} for both samples. The top panel shows the dollar value effects of one unit increase of position at different slots (1,5,10,20). The bottom panel shows the the dollar value effects of one unit increase of booking window at different days (1,5,10,20). 
\end{tablenotes}
\end{threeparttable}
\end{table} 

\begin{table}
\begin{threeparttable}[H]
\caption{The Biases of Position Effects in Dollar Values}
\label{tab:PositionBias}
\begin{tabular}{lccc}
\hline \hline
& (1) & (2) & (3) \\
& \text{Likelihood non-random}  & \text{Likelihood random} & \text{Likelihood non-random}  \\
& \text{vs. Likelihood random} &  \text{vs. PMR random} &  \text{vs. PMR non-random} \\ \hline
& & & \\ 
\textit{Mechanism} &  \text{endogeneity} &  \text{misspecification} &  \text{endogeneity} + \text{misspecification} \\ [0.5ex]
\textit{Direction} & \text{upward} & \text{upward} & \text{upward} \\ [0.5ex]
\textit{Magnitude} &&& \\
\quad \text{At slot 1} & 1.1692 & 2.9941 & 3.9371 \\
\quad \text{At slot 5} & 1.5987 & 3.0928 & 4.4612 \\
\quad \text{At slot 10} & 2.2239 & 3.2207 & 5.2090 \\
\quad \text{At slot 20} & 3.8319 & 3.4920 & 7.0775 \\
\hline \hline
\end{tabular}
\begin{tablenotes}[normal,flushleft]
\item \textit{Notes} This table summarizes the different mechanisms of position effect biases by comparing the estimates using different estimators and samples. For example, column (1) compares likelihood estimated position effects using the non-random ranking sample with the likelihood estimates using the random ranking sample (estimates obtained from Table \ref{tab:dollareffect}). It also shows the direction and magnitude of biases at different slots for each mechanism.
\end{tablenotes}
\end{threeparttable}
\end{table} 

Given the dollar value effects in Table \ref{tab:dollareffect}, we can also analyze the magnitude and direction of different mechanisms of biases in the position effects estimated using the likelihood approach. I summarize these in Table \ref{tab:PositionBias}. Column (1) compares the likelihood estimates of the non-random ranking sample with those of the random ranking sample. We find that the position endogeneity creates an upward biased estimate of the position effect of \$1.17 at slot 1 and \$3.83 at slot 20. The intuition is simple: under the expedia ranking algorithm, the higher ranked hotels are likely to have higher unobserved quality, so unobserved product quality exacerbates the effect of position on reservation utility. Thus, ignoring the endogenous unobserved quality leads to an upward biased estimate of the position effect. 

Column (2) compares the likelihood estimates with the PMR estimates from the random ranking sample. Since the positions are randomly generated in this sample, the only mechanism of bias is misspecification of match value distribution.\footnote{Price endogeneity can be another mechanism of bias, but it mainly causes bias in preference estimates rather than position estimates.} The misspecification causes an upward bias of the estimated position effects of \$2.99 at slot 1 and \$3.49 at slot 20. The upward bias exists because the true $G^{-1}$ function estimated from the data is steeper than the misspecified $G^{-1}$ function implied from the assumed N(0,1) match value distribution (as will be shown in Section \ref{sec:kaggle_ginv} Figure \ref{fig:Ginv}). Intuitively, if a hotel is moved further down the list page, then the associated decrease in reservation utility due to the decrease in the true $G^{-1}$ function would be wrongly attributed to the position effect, thus causing an upward bias. This is also consistent with the monte carlo simulation results in the previous section. 

Column (3) compares the likelihood estimates with the PMR estimates from the non-random ranking sample. In this sample, both position endogeneity and misspecification of match value distribution exist as mechanisms for bias. Since each of the mechanism causes an upward bias, the two mechanisms combined lead to an even larger upward bias. At slot 1, the misspecification mechanism contributes more to the total bias than the endogeneity mechanism, whereas at slot 20, the endogeneity mechanism contributes slightly more. The total bias is \$3.94 at slot 1 and \$7.08 at slot 20.

To put our results in the literature, \cite{ursu2018} and \cite{chungetal2019} both use the simulate likelihood apporach to estimate each of the four largest destinations of the Expedia data set. They use the random ranking sample to remove the endogeneity concern of the position variable, although the price endogeneity concern might still be present.\footnote{\cite{chungetal2019} proposes in Appendix F to deal with potential price endogeneity by constructing a proxy for product quality.} They assume match value is distributed as standard normal and abstract away the pre-search taste shock. Moreover, they model search cost by position and a constant term, whereas we replace the constant term with consumers' booking window.\footnote{Since each consumer's booking window is the same for all hotels, we can think of booking window as a consumer-specific constant term.} Therefore, we can not directly compare the PMR estimates in this paper to their estimates, due to different sampling and modeling assumptions; we only discuss here the range of the dollar values of their estimated position effect. \cite{ursu2018} finds position effects at slot 1 ranging from \$0.55 to \$3.19; \cite{chungetal2019} estimates range from \$0.95 to \$18.33. Our simulated likelihood results using the random ranking sample is \$3.11, which falls into their neighborhood of results. However, our PMR results at slot 1 ranges from \$0.11 to \$0.34, which is slightly smaller than their estimates. This is consistent with our analysis of bias above. 

Other studies of the position effect use different data sets and modeling assumptions. For example, \cite{chenyao2017} allow consumers to apply refinement tools and estimate position effect to be \$0.21, which is similar to our results. With sorting and filtering options, the position variable is unlikely to be endogeneous because it is determined by the observable filtering characteristics. That is, the hotel's position is unlikely to be correlated with its unobserved product quality once we condition on the observable filtering characteristics. They deal with potential price endogeneity by a control function approach. Since they are using a different data set with different consumers, whether misspecification of match value distribution exists is not clear. Due to the reasons mentioned above, their results are close to our PMR estimates. \cite{de2017} estimates the position effect to be from \$7.76 to \$35.15; \cite{koulayev2014} finds position effects ranging from \$2.93 to \$18.78; \cite{ghose2012surviving} finds the position effect to be \$6.24. The difference likely comes from potential position and price endogeneity, as well as different modeling assumptions of consumers' search processes and search costs.

\subsection{Match Value Distribution Estimates} \label{sec:kaggle_ginv}

As described above, in addition to estimating the preference and search cost parameters, I also nonparametrically estimate the $G^{-1}$ function using the estimator described in Section \ref{sec:nonparam}. Recall that the function $G$ is defined as in equation \eqref{eq:G}, so the estimated $G$ will inform us about the true match value distribution in the data. I use polynomials with degree 3 of the log of search cost as the basis. For the ease of computation, I assume that $G^{-1}$ is the same for all hotels. 

Figure \ref{fig:Ginv} shows the $G^{-1}$ estimated using the random ranking and non-random ranking sample (the solid black and grey lines in the figure). The estimated functions using both samples are quite close to each other, consistent with the fact that consumers are randomly assigned to different ranking algorithms (at least consumers in different samples don't have systematically very different match value distributions). The dashed black and grey lines are the implied $G^{-1}$ if we assume that match value is distributed as N(0,1). In order to compare the implied $G^{-1}$ (dashed) with the estimated $G^{-1}$ (solid), we normalize the implied $G^{-1}$ by the likelihood estimated price coefficient for each of the two samples, because we fix the price coefficient to be -1 when estimating the $G^{-1}$ function. Operationally, the normalization involves dividing the variance of match value distribution (which is 1) by the squared raw price coefficient estimated using simulated likelihood on each of the two samples, and then compute the associated $G^{-1}$ from the normalized match value distribution.\footnote{Note that normalization implies $\frac{u}{\abs{\beta_{\text{price}}}} = \frac{\delta}{\abs{\beta_{\text{price}}}} + \frac{\epsilon}{\abs{\beta_{\text{price}}}}$. Thus, if $\epsilon \sim$ N(0,1), then $\frac{\epsilon}{\abs{\beta_{\text{price}}}} \sim $ N(0, 1/$\beta_{\text{price}}^{2}$).} 

Comparing the solid with the dashed lines, we can see that the estimated $G^{-1}$ functions are steeper than the assumed $G^{-1}$ functions for both of the two samples. This suggests that the assumption of match value distribution following N(0,1) is very likely to be misspecified. This misspecification leads to an upward biased estimates of the dollar value of the position effect. The reason is: as the hotel's position increases, the associated decrease in reservation utility due to the decrease in the true $G^{-1}$ function would be wrongly attributed to the position effect.

\begin{figure}[hp]
\caption[]{Estimated $G^{-1}$ vs. Assumed $G^{-1}$ for Both Samples}
\label{fig:Ginv}
\centering
\begin{minipage}{.8\linewidth}
  \includegraphics[width=\linewidth]{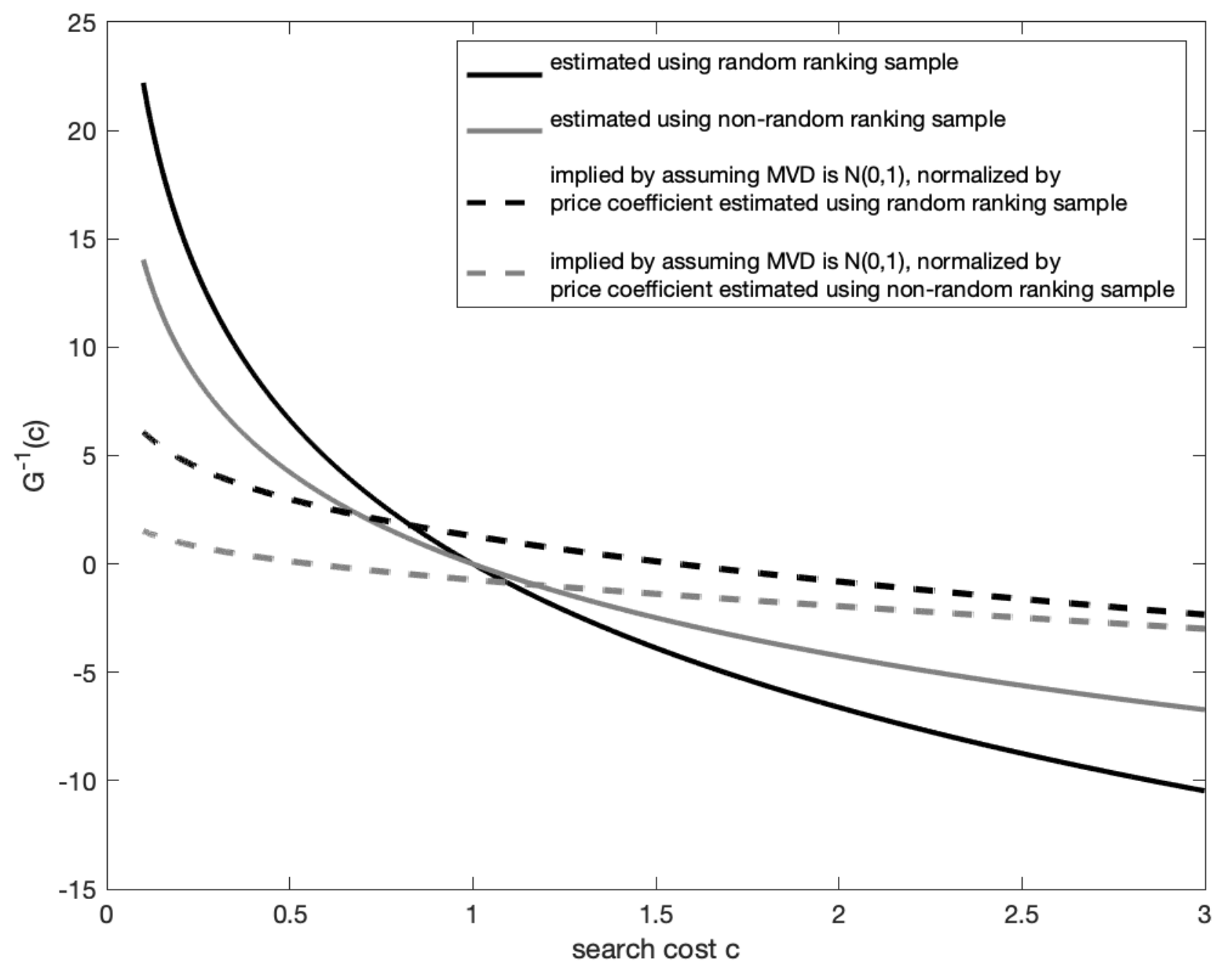} \\
\footnotesize
\emph{Notes} This figure plots the estimated $G^{-1}$ using the two samples of the Expedia data, as well as the $G^{-1}$ implied by assuming match value distribution to be N(0,1). The implied $G^{-1}$ is  after we normalize by the likelihood estimated price coefficient for each of the two samples.
\end{minipage}
\end{figure}

\section{Extensions} \label{sec:extension}

In this section, I will discuss some interesting extensions of the baseline model and estimator. Section \ref{sec:nonparam} shows how we can nonparametrically estimate the $G_j^{-1}$ function and use the estimates to test whether a candidate match value distribution is consistent with what the data implies. Section \ref{sec:zutil} deals with the situation where some search cost variables are also part of consumers' utility. Section \ref{sec:invar} describes how to estimate the preference for covariates that are same for all consumers. Section \ref{sec:hetero} extends the model to incorporate heterogeneous preference and search cost parameters and discusses estimation. 

\subsection{Nonparametric Estimation of $G_j^{-1}$} \label{sec:nonparam}

Before discussing how to estimate the inverse marginal benefit function $G_j^{-1}$ for each product $j$, I want to describe what the estimator should be if $G_j^{-1}$ is known. According to equation \eqref{eq:G}, if the researcher happens to know the true match value distribution for some reason, then she knows the $G_j^{-1}$ function for each product. The estimators introduced in Sections \ref{sec:nonsmooth} and \ref{sec:smooth} are still valid, but now she can take advantage of the fact that the match value distribution is known. For notational simplicity, define 
\begin{align*}
\tilde{h} (b,m; x_a, z_a) = G^{-1}_i ( \exp (z_{ai}' m )) - G^{-1}_j ( \exp (z_{aj}' m )) + x_{aij}' b
\end{align*}
It is then straightforward to see that we can just maximize the following objective for consistent estimates of the preference and search cost parameters:
\begin{align*}
\tilde{GQ}_{\mathcal{A}_{ij}} (b,m) =  {\abs{\mathcal{A}_{ij}} \choose 2}^{-1} \sum_{a \neq \tilde{a} \in \mathcal{A}_{ij}} 
\bigg[ & \mathds{1}\{ \tilde{h} (b,m; x_a, z_a) > \tilde{h} (b,m; x_{\tilde{a}}, z_{\tilde{a}}) \} \mathds{1}\{ S_{aij} > S_{\tilde{a}ij} \} \\
	  + & \mathds{1}\{ \tilde{h} (b,m; x_a, z_a) < \tilde{h} (b,m; x_{\tilde{a}}, z_{\tilde{a}}) \} \mathds{1}\{ S_{aij} < S_{\tilde{a}ij} \} \bigg]
\end{align*}

Now we turn to the more general case where the $G_j^{-1}$ function is unknown for all products. We assume that it can be well-approximated by linear sieves of degree $N$. More specifically, we approximate $G^{-1}_j(v)$ with $\sum_{n=1}^N \alpha_n^{j} \phi_n^{j} (v)$. Here $ \{ \phi_n^{j} \}_{n=1}^N$ are known basis functions and $ \{ \alpha_n^{j} \}_{n=1}^N$ are parameters to be estimated. For notational simplicity, define 
\begin{align} \label{eq:h}
h (b,m, \{ a_n^{i}, a_n^{j} \}_{n=1}^N; x_a, z_a) = \sum_{n=1}^N \left[ a_n^{i} \phi_n^{i} ( \exp( z_{ai}'m ) ) - a_n^{j} \phi_n^{j} ( \exp( z_{aj}'m ) ) \right] + x_{aij}'b
\end{align}
Then the objective to be maximized becomes
\begin{align*}
GQ_{\mathcal{A}_{ij}} & \left(b,m,\{ a_n^{i}, a_n^{j} \}_{n=1}^N \right) = {\abs{\mathcal{A}_{ij}} \choose 2}^{-1} \sum_{a \neq \tilde{a} \in \mathcal{A}_{ij}} \\
\bigg[ & \mathds{1}\{ h \left( b,m,\{ a_n^{i}, a_n^{j} \}_{n=1}^N; x_a, z_a \right) > h \left( b,m,\{ a_n^{i}, a_n^{j} \}_{n=1}^N; x_{\tilde{a}}, z_{\tilde{a}} \right) \} \mathds{1}\{ S_{aij} > S_{\tilde{a}ij} \} \\
+ &  \mathds{1}\{ h \left( b,m,\{ a_n^{i}, a_n^{j} \}_{n=1}^N; x_a, z_a \right) < h \left( b,m,\{ a_n^{i}, a_n^{j} \}_{n=1}^N; x_{\tilde{a}}, z_{\tilde{a}} \right) \} \mathds{1}\{ S_{aij} < S_{\tilde{a}ij} \} \bigg]
\end{align*}

The estimates $\{ \hat{a}_n^{i}, \hat{a}_n^{j} \}_{n=1}^N$ can then be used to test whether a candidate match value distribution is close to the true distribution. More specifically, for a product $j$, we can proceed in the following steps:
\begin{enumerate}
\item Propose a candidate match value distribution $\tilde{f}_j$.
\item Compute the $\tilde{G}^{-1}_{j}(v)$ associated with the proposed distribution from the definition:
\begin{align*}
\tilde{G}_{j}(w) = \int_w^{\infty} (\epsilon - w) \tilde{f}_{j}(\epsilon) d \epsilon
\end{align*}
\item Compute the estimated $\hat{G}^{-1}_{j} (v) =  \sum_{n=1}^N \hat{a}^{j}_n \phi^{j}_n (v)$. 
\item Reject $\tilde{f}_j$ if the error $\tilde{G}^{-1}_{j} (v) - \hat{G}^{-1}_{j}(v)$ is distributed very differently from zero.
\end{enumerate}
The logical argument of this testing procedure is quite straightforward. Let $G^{-1}_{j}$ denote the function associated with the true distribution $f_j$. From equation \eqref{eq:G}, we know that if $\tilde{f}_j = f_j$, then we must have that the associated $\tilde{G}^{-1}_j$ and $G^{-1}_j$ should also be the same. The logical equivalence of this argument is: if $\tilde{G}^{-1}_j$ and $G^{-1}_j$ are not the same, then $\tilde{f}_j$ and $f_j$ are also not the same. To understand the error better, we can write
\begin{align*}
\text{error} = \tilde{G}^{-1}_{j} (v) - \hat{G}^{-1}_{j}(v) =  \underbrace{ \left( \tilde{G}^{-1}_{j}(v) - G^{-1}_{j} (v) \right)}_{\text{misspecification error}} + \underbrace{ \left( G^{-1}_{j} (v) - \hat{G}^{-1}_{j} (v) \right) }_{\text{approximation error}} 
\end{align*} 
The error in the sample to be tested on is the sum of the misspecification error and the approximation error. Note that the approximation error is the same for different proposed distributions $\tilde{f}_j$ because the approximation error only depends on the true distribution $f_j$. Therefore, if we have several candidate distributions, the distribution with the smallest sample error also has the smallest misspecification error. This process can serve as a robustness check on the assumptions of match value distribution made in most empirical applications. 

\subsection{If Search Cost Variables are Included in Utility} \label{sec:zutil}

In the main specification of the model, we assume that the consumer's search cost variables are excluded from utility. In this section, we allow some or all of the search cost variables to enter the consumer's utility directly. More specifically, the utility is
\begin{align*}
u_{aj} = x_{aj}' \beta +  z_{aj}^{u'} \beta_z + \xi_j + \nu_a + \eta_{aj} + \epsilon_{aj}  
\end{align*}
where $z_{aj}^{u}$ is a subset of the search cost variables $z_{aj}$ (they can also be the same). Search cost is still modeled as in equation \eqref{eq:cparam}. Then we can write the reservation utility equation as
\begin{align*} 
 r_{aj} = G^{-1}_{j} \left(\exp \left( z_{aj}' \gamma \right) \right) + x_{aj}' \beta +  z_{aj}^{u'} \beta_z + \xi_j +\nu_a + \eta_{aj}
\end{align*}	
The outcome variable $S_{aij}$ is defined the same as in equation \eqref{eq:Sdef}. The estimator in this case can be constructed similarly as in Section \ref{sec:nonparam} by nonparametrically estimating the $G^{-1}_{j}$ function. First, define
\begin{align*}
hz (b,b_z,m, \{ a_n^{i}, a_n^{j} \}_{n=1}^N; x_a, z_a) = \sum_{n=1}^N \left[ a_n^{i} \phi_n^{i} ( \exp( z_{ai}'m ) ) - a_n^{j} \phi_n^{j} ( \exp( z_{aj}'m ) ) \right] + x_{aij}'b + z_{aij}^{u'} b_z
\end{align*}
where $z_{aij}^{u} = z_{ai}^{u} - z_{aj}^{u}$. The objective to be maximized is therefore
\begin{align*}
ZQ_{\mathcal{A}_{ij}} & \left(b,b_z,m,\{ a_n^{i}, a_n^{j} \}_{n=1}^N \right) = {\abs{\mathcal{A}_{ij}} \choose 2}^{-1} \sum_{a \neq \tilde{a} \in \mathcal{A}_{ij}} \\
\bigg[ & \mathds{1}\{ hz \left( b,b_z,m,\{ a_n^{i}, a_n^{j} \}_{n=1}^N; x_a, z_a \right) > hz \left( b,b_z,m,\{ a_n^{i}, a_n^{j} \}_{n=1}^N; x_{\tilde{a}}, z_{\tilde{a}} \right) \} \mathds{1}\{ S_{aij} > S_{\tilde{a}ij} \} \\
+ &  \mathds{1}\{ hz \left( b,b_z,m,\{ a_n^{i}, a_n^{j} \}_{n=1}^N; x_a, z_a \right) < hz \left( b,b_z,m,\{ a_n^{i}, a_n^{j} \}_{n=1}^N; x_{\tilde{a}}, z_{\tilde{a}} \right) \} \mathds{1}\{ S_{aij} < S_{\tilde{a}ij} \} \bigg]
\end{align*}
Intuitive, $\beta_z$ is separately identified from $\gamma$ because it enters linearly into the outcome equation. 

\subsection{Estimate Preference for Consumer-Invariant Characteristics} \label{sec:invar}

This section extends the baseline model to incorporate observed characteristics that are same for all consumers or search impressions. To be more specific, let utility be 
\begin{align*}
u_{aj} = x_{aj}' \beta +  \bar{x}_{j}' \beta_{\bar{x}} + \xi_j + \nu_a + \eta_{aj} + \epsilon_{aj}  
\end{align*}
where $\bar{x}_j$ represents some observed characteristics of product $j$ that are consumer invariant. The consumer's search cost is still modeled as in equation \eqref{eq:cparam}. Therefore, the reservation utility can now be written as
\begin{align*} 
 r_{aj} = G^{-1}_{j} \left(\exp \left( z_{aj}' \gamma \right) \right) + x_{aj}' \beta + \bar{x}_{j}' \beta_{\bar{x}} + \xi_j +\nu_a + \eta_{aj}
\end{align*}	
The outcome variable $S_{aij}$ is defined the same as in equation \eqref{eq:Sdef}. Recall that the endogeneity issue arises from the fact that the unobserved product quality $\xi_j$ can be correlated with observed characteristics $x_{aj}$ and search costs variables $z_{aj}$. We write  $x_{aj} = \left( x_{aj}^e, \,  x_{aj}^{-e} \right)$ and $z_{aj} = \left( z_{aj}^e, \,  z_{aj}^{-e} \right)$, where $(x_{aj}^e, z_{aj}^e)$ are the endogenous components and $(x_{aj}^{-e}, z_{aj}^{-e})$ are the exogenous components.

To estimate such a model, we propose a two-step estimator. In the first step, we use the estimator described in Section \ref{sec:nonparam} to get consistent estimates of the preference for consumer-varying characteristics $\beta$ and the search cost parameter $\gamma$, as well as the nonparametric approximation of the $G^{-1}$ function. Denote these estimates by $\hat{\theta} = ( \hat{\beta}, \hat{\gamma}, \{ \hat{\alpha}_n^{i}, \hat{\alpha}_n^{j} \}_{n=1}^N )$. Now the only parameter left to be estimated is $\beta_{\bar{x}}$. Because $\bar{x}_j$ does not vary across consumers, it can not be identified by comparing outcome across consumers. However, we can the exploit variation within consumer to estimate $\beta_{\bar{x}}$. In order to do that, we will impose an additional assumption.
\begin{assumption} \label{as:endvar}
The researcher observes consumer $a$'s outcome variable $S_{aij}$ for any product pair $(i,j) \in \mathcal{J}_{a}$. If $x_{ai}^e = x_{aj}^e$ and $z_{ai}^e = z_{aj}^e$, then $\xi_i$ and $\xi_j$ are iid. 
\end{assumption}	
Thus, $\mathcal{J}_a$ is the set of product pairs whose outcome variables are observed by the researcher. For example, if the researcher has data on the consideration set of consumer $a$, then we have $\mathcal{J}_a = \{ (i,j) \vert a \text{ searches } i \text{ but not } j \}$. This is because the reservation utility of any searched product must be higher than that of any unsearched product. In this case, if the total number of products is large enough, then $\abs{\mathcal{J}_a}$ will be much larger since we can take all possible combinations of any searched and unsearched product. 

The second step of the estimation involves maximizing the following objective
\begin{align*}
& XQ_{\mathcal{J}_a}  \left( b_{\bar{x}} ; \, \hat{\theta} \right) = \abs{\mathcal{J}_{a}}^{-1} \sum_{ (i, j) \in \mathcal{J}_{a}} \mathds{1} \big \{ x_{ai}^e = x_{aj}^e, \, z_{ai}^e = z_{aj}^e  \big \} \\
& \bigg[ S_{aij} \,  \mathds{1} \big \{h ( \hat{\theta} ; x_a, z_a ) + \left( \bar{x}_{i}' - \bar{x}_{j}' \right) b_{\bar{x}} > 0  \} 
+ ( 1-S_{aij} ) \,  \mathds{1} \big \{h ( \hat{\theta}; x_a, z_a ) + \left( \bar{x}_{i}' - \bar{x}_{j}' \right) b_{\bar{x}} < 0  \} \bigg]
\end{align*}
where $h$ is defined as in equation \eqref{eq:h} in Section \ref{sec:nonparam}. The intuition is simple: if product $i$'s non-error part of the reservation utility is larger than that of product $j$, then $S_{aij}$ is likely to be 1 in expectation; if smaller, then $S_{aij}$ is likely to be 0. This is because the error parts are distributed iid for both products conditional on endogenous components being equal by Assumption \ref{as:endvar} and the assumption that the taste shock $\eta_{aij}$ is also iid. Moreover, the smoothing techniques described in Section \ref{sec:smooth} are also applicable here to deal with potentially continuously distributed endogenous components. Additionally, we can sum the objective above across consumers to increase efficiency.

\subsection{Heterogeneous Preference and Search Cost} \label{sec:hetero}

In this section, we discuss how to incorporate heterogeneous preference and search cost parameters into the baseline model. The utility and search cost are now modeled as:
\begin{align*}
u_{aj} &= x_{aj}' \beta_a + \xi_j + \nu_a + \eta_{aj} + \epsilon_{aj}  \\
c_{aj} &= \exp \left( z_{aj}' \gamma_a \right)
\end{align*}
Now $\beta_a$ and $\gamma_a$ are consumer specific preference and search cost parameters to be estimated. If we have enough computation power, then we can use the same estimator as in Section \ref{sec:nonparam} to estimate the preference and search cost parameters for each consumer. We would then be able to recover consumer level estimates rather than a distribution of parameters. However, we might lack the computational resources to estimates these parameters jointly. Now suppose the researcher also has some data on the consumers' demographics that could affect their preference and search costs. Let $\beta_a = \bar{\beta} + d_a^{x'} \beta^{d}$ and $\gamma_a = \bar{\gamma} + d_a^{z'} \gamma^{d}$, where $d_a^x$ and $d_a^z$ are demographic variables of consumer $a$. Define the following
\begin{align*} 
h_{a} (\bar{b}, b^d, & \bar{m}, m^d, \{ a_n^{i}, a_n^{j} \}_{n=1}^N) = \\ 
& \sum_{n=1}^N \left[ a_n^{i} \phi_n^{i} ( \exp( z_{ai}' (\bar{m} + d_a^{z'} m^d) ) ) - a_n^{j} \phi_n^{j} ( \exp( z_{aj}' (\bar{m} + d_a^{z'} m^d) ) ) \right] + x_{aij}' (\bar{b} + d_a^{x'} b^d)
\end{align*}
The estimator in this case maximizes the following objective
\begin{align*}
DQ_{\mathcal{A}_{ij}} & \left(\bar{b}, b^d, \bar{m}, m^d,\{ a_n^{i}, a_n^{j} \}_{n=1}^N \right) = {\abs{\mathcal{A}_{ij}} \choose 2}^{-1} \sum_{a \neq \tilde{a} \in \mathcal{A}_{ij}} \\
\bigg[ & \mathds{1}\{ h_{a} ( \bar{b}, b^d, \bar{m}, m^d, \{ a_n^{i}, a_n^{j} \}_{n=1}^N ) > h_{\tilde{a}} (  \bar{b}, b^d, \bar{m}, m^d, \{ a_n^{i}, a_n^{j} \}_{n=1}^N ) \} \mathds{1}\{ S_{aij} > S_{\tilde{a}ij} \} \\
+ & \mathds{1}\{ h_{a} ( \bar{b}, b^d, \bar{m}, m^d, \{ a_n^{i}, a_n^{j} \}_{n=1}^N ) < h_{\tilde{a}} (  \bar{b}, b^d, \bar{m}, m^d, \{ a_n^{i}, a_n^{j} \}_{n=1}^N ) \} \mathds{1}\{ S_{aij} < S_{\tilde{a}ij} \} \bigg]
\end{align*}
More generally, we can allow for $\beta_a = \beta(d_a^x)$ and $\gamma_a = \gamma (d_a^z)$ where $\beta(.)$ and $\gamma(.)$ are arbitrary functions to be estimated.

\section{Conclusion}

This paper presents a sequential search model that allows for endogenous unobserved product quality and researchers' misspecification of match value distribution. I propose a new pairwise maximum rank estimator that consistently estimates both consumers' preference and search cost parameters in the more general model. The PMR estimator is easily implementable and flexible in the sense that it can be adapted to various empirical scenarios. 

The main message is that the PMR estimator enables us to apply sequential search models to a much wider range of interesting and important empirical contexts which we weren't able to due to concerns of endogeneity, misspecification of match value distribution, or data and computational limitation. For example, with the PMR estimator, we can now study consumers' offline choices where it is difficult to observe their full consideration sets; school or housing choice (and other economically important contexts) where unobserved product quality most likely exists and causes endogeneity. We can also test for consumers' actual beliefs about the distribution of their match value, which can serve as a robustness check for assumptions of match value distributions in empirical applications. The PMR estimator can only be applied in sequential search models. So exploring how it can be adapted to simultaneous search models is a good avenue for future research. Another interesting question is to derive more theoretical results on inference methods of the PMR estimator.

\clearpage
\singlespace
\bibliography{ref}
\bibliographystyle{chicago}
\clearpage

\begin{appendices}
\normalsize

\section{Proof of Lemma \ref{lem:inv}} \label{app:lemma1pf}

Recall that $G_j : \left[ \underline{w}, \overline{w} \right] \rightarrow C$ where $C \subseteq \mathbb{R}$ is the image.
\begin{align*} 
G_j(w) = \int_{w}^{\infty} ( \epsilon-w) f_{j}(\epsilon) d \epsilon
\end{align*}
First, we want to show $G_j$ is strictly decreasing. Take any $w_l < w_h \in  \left[ \underline{w}, \overline{w} \right]$, define
\begin{align*} 
\tilde{G}_j(w_l, w_h) = \int_{w_h}^{\infty} ( \epsilon - w_l ) f_{j}(\epsilon) d \epsilon
\end{align*}
In order to show $G_j$ is strictly decreasing, we just need to show $G_j(w_l) > G_j(w_h)$. First we show $G_j(w_l) > \tilde{G}_j(w_l, w_h)$.
\begin{align*}
G_j(w_l) &= \int_{w_l}^{\infty} ( \epsilon - w_l ) f_{j}(\epsilon) d \epsilon \\
&= \int_{w_l}^{w_h} ( \epsilon - w_l ) f_{j}(\epsilon) d \epsilon + \int_{w_h}^{\infty} ( \epsilon - w_l ) f_{j}(\epsilon) d \epsilon \\
&= \int_{w_l}^{w_h} ( \epsilon - w_l ) f_{j}(\epsilon) d \epsilon + \tilde{G}_j(w_l, w_h) 
\end{align*}
Since $f_{j}$ is strictly positive on $\left[ \underline{w}, \overline{w} \right]$ and $\epsilon - w_l$ is strictly positive for any $\epsilon \in (w_l, w_h]$, we know that $ \int_{w_l}^{w_h} ( \epsilon - w_l ) f_{j}(\epsilon) d \epsilon$ is also strictly positive. Thus, we have $G_j(w_l) > \tilde{G}_j(w_l, w_h)$.

Now we want to show $\tilde{G}_j(w_l, w_h) \geq G_j(w_h)$. 
\begin{align*}
\tilde{G}_j(w_l, w_h) &= \int_{w_h}^{\infty} ( \epsilon - w_l ) f_{j}(\epsilon) d \epsilon \\
& \geq \int_{w_h}^{\infty} ( \epsilon - w_h ) f_{j}(\epsilon) d \epsilon \\
&=  G_j(w_h)
\end{align*}
The inequality holds because $ \epsilon - w_l >   \epsilon - w_h \geq 0$ for any $ \epsilon \geq w_h > w_l$. Thus, we have $G_j(w_l) > \tilde{G}_j(w_l, w_h) \geq G_j(w_h)$.

Because $G_j$ is strictly decreasing, it is bijective and therefore invertible. Next, we want to show $G_j^{-1}$ is also strictly decreasing. For any $c_l < c_h \in C$, there exist $w_l$ and $w_h$ such that $G_j(w_l)=c_l$ and $G_j(w_h)=c_h$ because $G_j$ is bijective. Suppose by way of contradiction that $w_l \leq w_h$, then we have $G_j(w_l) \geq G_j(w_h)$ because $G_j$ is strictly decreasing. This implies that $c_l \geq c_h$, which is a contradiction. Thus, we must have that $w_l > w_h$, so $G_j^{-1}$ is strictly decreasing.

\section{Proof of Theorem \ref{thm:nonsmooth}} \label{app:thm1pf}

Consider the following probability limits of the objective of the ideal estimator:
\begin{align*} 
Q_1(b)  = \mathbb{E}  \bigg[ 
	& \mathds{1}\{z_{ai} = z_{\tilde{a}i}, z_{aj} = z_{\tilde{a}j} \} \mathds{1}\{x_{aij}'b > x_{\tilde{a}ij}'b\} \, \mathbb{E} \left[  \mathds{1}\{S_{aij} > S_{\tilde{a}ij}\} \vert x,z \right] \\
 + & \mathds{1}\{z_{ai} = z_{\tilde{a}i}, z_{aj} = z_{\tilde{a}j} \}  \mathds{1}\{x_{aij}'b < x_{\tilde{a}ij}'b\} \, \mathbb{E} \left[  \mathds{1}\{S_{aij} < S_{\tilde{a}ij}\} \vert x, z \right]   \bigg] \\
 Q_2(m)  = \mathbb{E}  \bigg[ 
	&  \mathds{1}\{z_{ai} = z_{\tilde{a}i}, x_{aij} = x_{\tilde{a}ij} \} \mathds{1}\{z_{aj}'m > z_{\tilde{a}j}'m\} \, \mathbb{E} \left[  \mathds{1}\{S_{aij} > S_{\tilde{a}ij}\} \vert x,z \right] \\
 + &  \mathds{1}\{z_{ai} = z_{\tilde{a}i}, x_{aij} = x_{\tilde{a}ij} \} \mathds{1}\{z_{aj}'m < z_{\tilde{a}j}'m\}  \, \mathbb{E} \left[  \mathds{1}\{S_{aij} < S_{\tilde{a}ij}\} \vert x, z \right]   \bigg] \\
  Q_3(m)  = \mathbb{E}  \bigg[ 
	&  \mathds{1}\{z_{aj} = z_{\tilde{a}j}, x_{aij} = x_{\tilde{a}ij} \} \mathds{1}\{z_{ai}'m < z_{\tilde{a}i}'m\} \, \mathbb{E} \left[  \mathds{1}\{S_{aij} > S_{\tilde{a}ij}\} \vert x,z \right] \\
 + &  \mathds{1}\{z_{aj} = z_{\tilde{a}j}, x_{aij} = x_{\tilde{a}ij} \} \mathds{1}\{z_{ai}'m > z_{\tilde{a}i}'m\}  \, \mathbb{E} \left[  \mathds{1}\{S_{aij} < S_{\tilde{a}ij}\} \vert x, z \right]   \bigg] \\
Q(b,m) = \quad & Q_1(b) + Q_2 (m) +Q_3 (m)	
\end{align*}

To show strong consistency, we need to show the following conditions hold: \\
(i) $Q(b,m)$ is uniquely maximized at $(\beta^{*}, \gamma^{*})$. \\
(ii) $\Theta_{\rho}$ is compact. \\
(iii) $Q(b,m)$ is continuous. \\
(iv) $Q_{ \mathcal{A}_{ij}} (b,m)$ converges uniformly almost surely to $Q(b,m)$ \\ 
(i.e. $\sup_{(b,m) \in \Theta_{\rho}} \abs{Q_{ \mathcal{A}_{ij}} (b,m) - Q(b,m)} \overset{a.s.}{\to} 0$) 

\subsection*{Prove that condition (i) holds}

First we want to show $Q(b,m)$ is maximized at  $(\beta^{*}, \gamma^{*})$. To show maximization, we make the following Lemma \ref{lem:maxb} and \ref{lem:maxm}.

\begin{lemma} \label{lem:maxb}
For any $x_{aij}, x_{\tilde{a}ij}, z_{ai}, z_{\tilde{a}i}, z_{aj}, z_{\tilde{a}j}$, \\
if $x_{aij}' \beta^{*} > x_{\tilde{a}ij}' \beta^{*},  z_{ai} = z_{\tilde{a}i}, z_{aj} = z_{\tilde{a}j}$, then $\mathbb{E} \left[  \mathds{1}\{S_{aij} > S_{\tilde{a}ij}\} \vert x,z \right] > \mathbb{E} \left[  \mathds{1}\{S_{\tilde{a}ij} > S_{aij}\} \vert x,z \right]$; \\
if $x_{aij}' \beta^{*} < x_{\tilde{a}ij}' \beta^{*},  z_{ai} = z_{\tilde{a}i}, z_{aj} = z_{\tilde{a}j}$, then $\mathbb{E} \left[  \mathds{1}\{S_{aij} > S_{\tilde{a}ij}\} \vert x,z \right] < \mathbb{E} \left[  \mathds{1}\{S_{\tilde{a}ij} > S_{aij}\} \vert x,z \right]$.
\end{lemma}

\begin{proof}
First notice that $x_{aij}' \beta^{*} > x_{\tilde{a}ij}' \beta^{*}$ is equivalent to $x_{aij}' \beta > x_{\tilde{a}ij}' \beta$. Let $h \left( z_{ai}, z_{aj}, \gamma \right) = G_{i}^{-1} ( \exp (z_{ai}' \gamma )) - G_{j}^{-1} ( \exp (z_{aj}' \gamma ))$. Denote $\xi_{ij} = \xi_i - \xi_j$, $\eta_{aij} = \eta_{ai} - \eta_{aj}$ and $\eta_{\tilde{a}ij} = \eta_{\tilde{a}i} - \eta_{\tilde{a}j}$. Given any $x_{aij}' \beta > x_{\tilde{a}ij}' \beta, z_{ai} = z_{\tilde{a}i} = z_i, z_{aj} = z_{\tilde{a}j} = z_j, \xi_{ij}$, we have
\begin{align*}
&  \mathbb{E} \left[  \mathds{1}\{S_{aij} > S_{\tilde{a}ij} \} \vert x,z, \xi \right]  \\
= & \mathbb{E} \left[ \mathds{1}\{ \mathds{1}\{ h \left( z_{ai}, z_{aj}, \gamma \right) + x_{aij}' \beta + \xi_{ij} + \color{red} \eta_{aij} \color{black} >0 \}    >  \mathds{1}\{ h \left( z_{\tilde{a}i}, z_{\tilde{a}j}, \gamma \right)  + x_{\tilde{a}ij}' \beta + \xi_{ij} + \color{red} \eta_{\tilde{a}ij} \color{black}  >0 \}  \} \vert x,z, \xi \right] \\
 = & \mathbb{E} \left[  \mathds{1}\{ \mathds{1}\{ h \left( z_{\tilde{a}i}, z_{\tilde{a}j}, \gamma \right) + \color{blue} x_{aij}' \beta \color{black} + \xi_{ij} + \color{red} \eta_{\tilde{a}ij} \color{black}  >0 \} >  \mathds{1}\{h \left( z_{ai}, z_{aj}, \gamma \right) + x_{\tilde{a}ij}' \beta + \xi_{ij} +  \color{red} \eta_{aij} \color{black}   >0 \} \} \vert x,z, \xi \right] \\ 
>  & \mathbb{E} \left[  \mathds{1}\{ \mathds{1}\{h \left( z_{\tilde{a}i}, z_{\tilde{a}j}, \gamma \right) + \color{blue} x_{\tilde{a}ij}' \beta \color{black}+ \xi_{ij} + \eta_{\tilde{a}ij}  >0 \} >  \mathds{1}\{ h \left( z_{ai}, z_{aj}, \gamma \right) + \color{orange} x_{\tilde{a}ij}' \beta \color{black} + \xi_{ij} + \eta_{aij}  >0 \} \} \vert x,z, \xi \right] \\ 
\geq & \mathbb{E} \left[  \mathds{1}\{ \mathds{1}\{ h \left( z_{\tilde{a}i}, z_{\tilde{a}j}, \gamma \right) + x_{\tilde{a}ij}' \beta + \xi_{ij} + \eta_{\tilde{a}ij}  >0 \} >  \mathds{1}\{ h \left( z_{ai}, z_{aj}, \gamma \right) + \color{orange} x_{aij}' \beta \color{black} + \xi_{ij} + \eta_{aij}  >0 \} \} \vert x,z, \xi \right] \\ 
= & \mathbb{E} \left[  \mathds{1}\{S_{\tilde{a}ij} > S_{aij} \} \vert x,z, \xi \right]
\end{align*}

The first equality comes from the definition of $S$ as in equation \eqref{eq:S}. The second equality holds because $ \eta_{aj}$ is iid.

To show the first inequality holds, let $\Omega_a = \left( -\infty, - h \left( z_{i}, z_{j}, \gamma \right) - x_{\tilde{a}ij}' \beta - \xi_{ij} \right] \subseteq \mathbb{R}$ and \\ $\Omega_{\tilde{a}} = \left( - h \left( z_{i}, z_{j}, \gamma \right) - x_{aij}' \beta - \xi_{ij}, - h \left( z_{i}, z_{j}, \gamma \right) - x_{\tilde{a}ij}' \beta - \xi_{ij} \right] \subseteq \mathbb{R}$. Since $x_{aij}' \beta > x_{\tilde{a}ij}' \beta$, $\Omega_{\tilde{a}}$ is not empty. For any $\eta_{aij} \in \Omega_a$ and any $\eta_{\tilde{a}ij} \in \Omega_{\tilde{a}}$, we have 
\begin{eqnarray*}
&& \mathds{1}\{ h \left( z_{i}, z_{j}, \gamma \right) +  x_{\tilde{a}ij}' \beta + \xi_{ij} + \eta_{aij}  >0 \} = 0 \\
&& \mathds{1}\{ h \left( z_{i}, z_{j}, \gamma \right) + x_{\tilde{a}ij}' \beta + \xi_{ij} + \eta_{\tilde{a}ij}  >0 \} = 0 \\
&& \mathds{1}\{ h \left( z_{i}, z_{j}, \gamma \right) + x_{aij}' \beta + \xi_{ij} + \eta_{\tilde{a}ij}  >0 \} = 1 
\end{eqnarray*}
These imply that $ \forall \, \eta_{aij} \in \Omega_a$ and $ \forall \, \eta_{\tilde{a}ij} \in \Omega_{\tilde{a}}$, 
\begin{align*}
 & \mathds{1}\{ \mathds{1}\{ h \left( z_{i}, z_{j}, \gamma \right) +  x_{aij}' \beta + \xi_{ij} +  \eta_{\tilde{a}ij}  >0 \} >  \mathds{1}\{ h \left( z_{i}, z_{j}, \gamma \right)+ x_{\tilde{a}ij}' \beta + \xi_{ij} +  \eta_{aij}   >0 \} \\
> & \mathds{1}\{ \mathds{1}\{ h \left( z_{i}, z_{j}, \gamma \right)+ x_{\tilde{a}ij}' \beta + \xi_{ij} + \eta_{\tilde{a}ij}  >0 \} >  \mathds{1}\{ h \left( z_{i}, z_{j}, \gamma \right) +  x_{\tilde{a}ij}' \beta + \xi_{ij} + \eta_{aij}  >0 \} \} 
\end{align*} 
Since $\eta_{aij}$ and $\eta_{\tilde{a}ij}$ have almost everywhere positive density on $\mathbb{R}$, the event $\Omega_a$ and $\Omega_{\tilde{a}}$ happen with positive probability. Therefore, the first inequality holds. The second inequality holds because  $x_{aij}' \beta > x_{\tilde{a}ij}' \beta$. The last equality holds from the definition of $S$. 

Since $\mathbb{E} \left[  \mathds{1}\{S_{aij} > S_{\tilde{a}ij} \} \vert x,z, \xi \right] > \mathbb{E} \left[  \mathds{1}\{S_{\tilde{a}ij} > S_{aij} \} \vert x,z, \xi \right]$ holds for all $\xi$, we have that \\ $\mathbb{E} \left[  \mathds{1}\{S_{aij} > S_{\tilde{a}ij} \} \vert x,z \right] > \mathbb{E} \left[  \mathds{1}\{S_{\tilde{a}ij} > S_{aij} \} \vert x,z \right]$. 

Similar argument will prove the case with $x_{aij}' \beta^{*} < x_{\tilde{a}ij}' \beta^{*}, z_{ai} = z_{\tilde{a}i}, z_{aj} = z_{\tilde{a}j}$.
\end{proof}

\begin{lemma} \label{lem:maxm}
For any $x_{aij}, x_{\tilde{a}ij}, z_{ai}, z_{\tilde{a}i}, z_{aj}, z_{\tilde{a}j}$, \\
if $z_{aj}' \gamma^{*} > z_{\tilde{a}j}' \gamma^{*}, x_{aij} = x_{\tilde{a}ij},  z_{ai} = z_{\tilde{a}i}$, then $\mathbb{E} \left[  \mathds{1}\{S_{aij} > S_{\tilde{a}ij}\} \vert x,z \right] > \mathbb{E} \left[  \mathds{1}\{S_{\tilde{a}ij} > S_{aij}\} \vert x,z \right]$; \\
if $z_{aj}' \gamma^{*} < z_{\tilde{a}j}' \gamma^{*}, x_{aij} = x_{\tilde{a}ij},  z_{ai} = z_{\tilde{a}i}$, then $\mathbb{E} \left[  \mathds{1}\{S_{aij} > S_{\tilde{a}ij}\} \vert x,z \right] < \mathbb{E} \left[  \mathds{1}\{S_{\tilde{a}ij} > S_{aij}\} \vert x,z \right]$. \\
if $z_{ai}' \gamma^{*} > z_{\tilde{a}i}' \gamma^{*}, x_{aij} = x_{\tilde{a}ij},  z_{aj} = z_{\tilde{a}j}$, then $\mathbb{E} \left[  \mathds{1}\{S_{aij} > S_{\tilde{a}ij}\} \vert x,z \right] < \mathbb{E} \left[  \mathds{1}\{S_{\tilde{a}ij} > S_{aij}\} \vert x,z \right]$; \\
if $z_{ai}' \gamma^{*} < z_{\tilde{a}i}' \gamma^{*}, x_{aij} = x_{\tilde{a}ij},  z_{aj} = z_{\tilde{a}j}$, then $\mathbb{E} \left[  \mathds{1}\{S_{aij} > S_{\tilde{a}ij}\} \vert x,z \right] > \mathbb{E} \left[  \mathds{1}\{S_{\tilde{a}ij} > S_{aij}\} \vert x,z \right]$. \\
\end{lemma}

\begin{proof}
Similar argument as in the proof of Lemma \ref{lem:maxb}.
\end{proof}

Lemma \ref{lem:maxb} implies that $Q_1(b)$ is maximized at $\beta^{*}$. Similarly, Lemma \ref{lem:maxm} implies that $\gamma^{*}$ maximizes both $Q_2(m)$ and $Q_3 (m)$. Thus, $(\beta^{*}, \gamma^{*})$ maximizes $Q(b,m)$.  

Next we want to show that $Q(b,m)$ is maximized uniquely at $(\beta^{*}, \gamma^{*})$. We first show that $\beta^{*}$ uniquely maximizes $Q_1(b)$. 

Let $\mathcal{B}_{\rho} = \{ b: \abs{b^1} \geq \rho, \, \norm{b} = 1 \}$. For any $b \in \mathcal{B}_{\rho}$ with $b \neq \beta^{*}$, define 
\begin{align*}
X(b, \beta^{*}) = & \, \{ x_{aij}, x_{\tilde{a}ij} \vert x_{aij}' b > x_{\tilde{a}ij}' b,  x_{aij}' \beta^{*} < x_{\tilde{a}ij}' \beta^{*}, z_{ai} = z_{\tilde{a}i}, z_{ai} = z_{\tilde{a}i} \} \\
\cup & \, \{ x_{aij}, x_{\tilde{a}ij} \vert x_{aij}' b < x_{\tilde{a}ij}' b,  x_{aij}' \beta^{*} > x_{\tilde{a}ij}' \beta^{*}, z_{ai} = z_{\tilde{a}i}, z_{ai} = z_{\tilde{a}i} \} 
\end{align*} 
First consider the case with $b^1 >0$ and $\beta^{*1} > 0$. We can rewrite $X(b, \beta^{*})$ as the following:
\begin{align*}
X(b, \beta^{*}) = & \, \{ \tilde{x}_{aij}, \tilde{x}_{\tilde{a}ij} \vert \left( \tilde{x}_{aij}' - \tilde{x}_{\tilde{a}ij}' \right) \tilde{b} / b^{1} > x_{\tilde{a}ij}^1 - x_{aij}^1 >  \left( \tilde{x}_{aij}' - \tilde{x}_{\tilde{a}ij}' \right) \tilde{\beta}^{*} / \beta^{*1}, z_{ai} = z_{\tilde{a}i}, z_{ai} = z_{\tilde{a}i} \} \\
\cup & \, \{ \tilde{x}_{aij}, \tilde{x}_{\tilde{a}ij} \vert \left( \tilde{x}_{aij}' - \tilde{x}_{\tilde{a}ij}' \right) \tilde{b} / b^{1} < x_{\tilde{a}ij}^1 - x_{aij}^1 <  \left( \tilde{x}_{aij}' - \tilde{x}_{\tilde{a}ij}' \right) \tilde{\beta}^{*} / \beta^{*1}, z_{ai} = z_{\tilde{a}i}, z_{ai} = z_{\tilde{a}i} \}
\end{align*} 
We also know that $ \left( \tilde{x}_{aij}' - \tilde{x}_{\tilde{a}ij}' \right) \tilde{b} / b^{1} =  \left( \tilde{x}_{aij}' - \tilde{x}_{\tilde{a}ij}' \right) \tilde{\beta}^{*} / \beta^{*1}$ for all $ \tilde{x}_{aij}$ and $\tilde{x}_{\tilde{a}ij}$ does not happen because of assumption \ref{as:support} (a) and $\tilde{b} / b^{1} \neq \tilde{\beta}^{*} / \beta^{*1}$. The latter is true because suppose, by way of contradiction, that $\tilde{b} / b^{1} = \tilde{\beta}^{*} / \beta^{*1}$, then $b = \left( b^{1} / \beta^{*1} \right) \beta^{*}$. This implies that $\norm{b} = \abs{b^{1} / \beta^{*1}} * \norm{\beta^{*}}$. Since $\norm{b} = \norm{\beta^{*}}=1$ by assumption, then $b^{1} / \beta^{*1} = 1$. But because $b \neq \beta^{*}$, we know that $ b^{1} / \beta^{*1} \neq 1$, which is a contradiction.

Assumption \ref{as:support} (b) makes sure that $X(b, \beta^{*})$ has positive measure for any $b \in \mathcal{B}_{\rho}$ with $b \neq \beta^{*}$. Since $X(b, \beta^{*})$ has positive probability, then $b$ will make incorrect predictions with positive probability as compared to the true date generating process. Therefore, $Q_1(b)$ will be lower than $Q_1(\beta^{*})$. 

Next consider the case with $b^1 < 0$ and $\beta^{*1} > 0$. We can rewrite $X(b, \beta^{*})$ as:
\begin{align*}
& X(b, \beta^{*}) \\
= & \, \{ \tilde{x}_{aij}, \tilde{x}_{\tilde{a}ij} \vert x_{\tilde{a}ij}^1 - x_{aij}^1 >  \left( \tilde{x}_{aij}' - \tilde{x}_{\tilde{a}ij}' \right) \tilde{\beta}^{*} / \beta^{*1}, \, x_{\tilde{a}ij}^1 - x_{aij}^1 > \left( \tilde{x}_{aij}' - \tilde{x}_{\tilde{a}ij}' \right) \tilde{b} / b^{1}, \, z_{ai} = z_{\tilde{a}i}, z_{ai} = z_{\tilde{a}i} \} \\
\cup & \, \{ \tilde{x}_{aij}, \tilde{x}_{\tilde{a}ij} \vert x_{\tilde{a}ij}^1 - x_{aij}^1 <  \left( \tilde{x}_{aij}' - \tilde{x}_{\tilde{a}ij}' \right) \tilde{\beta}^{*} / \beta^{*1}, \, x_{\tilde{a}ij}^1 - x_{aij}^1 < \left( \tilde{x}_{aij}' - \tilde{x}_{\tilde{a}ij}' \right) \tilde{b} / b^{1}, \, z_{ai} = z_{\tilde{a}i}, z_{ai} = z_{\tilde{a}i} \} \\
= & \, \{ \tilde{x}_{aij}, \tilde{x}_{\tilde{a}ij} \vert x_{\tilde{a}ij}^1 - x_{aij}^1 > \max \{ \left( \tilde{x}_{aij}' - \tilde{x}_{\tilde{a}ij}' \right) \tilde{\beta}^{*} / \beta^{*1}, \, \left( \tilde{x}_{aij}' - \tilde{x}_{\tilde{a}ij}' \right) \tilde{b} / b^{1} \}, \, z_{ai} = z_{\tilde{a}i}, z_{ai} = z_{\tilde{a}i} \} \\
\cup & \, \{ \tilde{x}_{aij}, \tilde{x}_{\tilde{a}ij} \vert x_{\tilde{a}ij}^1 - x_{aij}^1 < \min \{ \left( \tilde{x}_{aij}' - \tilde{x}_{\tilde{a}ij}' \right) \tilde{\beta}^{*} / \beta^{*1}, \, \left( \tilde{x}_{aij}' - \tilde{x}_{\tilde{a}ij}' \right) \tilde{b} / b^{1} \}, \, z_{ai} = z_{\tilde{a}i}, z_{ai} = z_{\tilde{a}i} \} 
\end{align*} 
Again, assumption \ref{as:support} guarantees that $X(b,\beta)$ has positive probability, which implies that $Q_1(b)$ will be lower than $Q_1(\beta^{*})$. Similar argument works for the case with $\beta^{*1} < 0$.

We have shown that for any $b \in \mathcal{B}_{\rho}$ with $b \neq \beta^{*}$, we have $Q_1(b) < Q_1 (\beta^{*})$. This means that $\beta^{*}$ uniquely maximizes $Q_1$. Similar arguments will show that $\gamma^{*}$ maximizes both $Q_2(m)+Q_3(m)$ uniquely. Therefore, $(\beta^{*}, \gamma^{*})$ maximizes $Q(b, m)$ uniquely. 

\subsection*{Prove that conditions (ii)-(iv) hold}

Condition (ii) holds by definition. Condition (iii) follows from \cite{manski1985} Lemma 5. Condition (iv) can be proved analogously to \cite{han1987}'s proof of uniform convergence.

\section{Proof of Theorem \ref{thm:smooth}} \label{app:thm2pf}

To show strong consistency, we need to show: \\
(i) $Q(b,m)$ is uniquely maximized at $(\beta^{*}, \gamma^{*})$. \\
(ii) $\Theta_{\rho}$ is compact. \\
(iii) $Q(b,m)$ is continuous. \\
(iv) $SQ_{ \mathcal{A}_{ij}} (b,m; \sigma_{\mathcal{A}_{ij}})$ converges uniformly almost surely to $Q(b,m)$ \\ 
(i.e. $\sup_{(b,m) \in \Theta_{\rho}} \abs{SQ_{ \mathcal{A}_{ij}} (b,m; \sigma_{\mathcal{A}_{ij}}) - Q(b,m)} \overset{a.s.}{\to} 0$) 

Conditions (i) - (iii) are already proven to hold from the proof of theorem \ref{thm:nonsmooth}. So we only need to verify condition (iv).

To show that $SQ_{ \mathcal{A}_{ij}} (b,m; \sigma_{\mathcal{A}_{ij}})$ converges uniformly almost surely to $Q(b,m)$, we just need to show   that $SQ_{ \mathcal{A}_{ij}} (b,m; \sigma_{\mathcal{A}_{ij}})$ converges uniformly almost surely to $Q_{ \mathcal{A}_{ij}} (b,m; \sigma_{\mathcal{A}_{ij}})$, because we have already shown $Q_{ \mathcal{A}_{ij}} (b,m; \sigma_{\mathcal{A}_{ij}})$ converges uniformly almost surely to $Q(b,m)$ in the proof of theorem \ref{thm:nonsmooth}.

Let
\begin{align*}
SQ_{1, \mathcal{A}_{ij}} (b) = {\abs{\mathcal{A}_{ij}} \choose 2}^{-1} \sum_{a \neq \tilde{a} \in \mathcal{A}_{ij}} 
    K \bigg(  \scriptsize  \frac{  \norm{ \begin{bmatrix}
	z_{ai} \\
	z_{aj}
	\end{bmatrix} - \begin{bmatrix}
	z_{\tilde{a}i} \\
	z_{\tilde{a}j}
	\end{bmatrix} }}  {\sigma_{\mathcal{A}_{ij}}} \bigg) 
	\normalsize
	 \bigg( & \mathds{1}\{x_{aij}'b > x_{\tilde{a}ij}'b\} \mathds{1}\{S_{aij} > S_{\tilde{a}ij}\} \\
	      +   & \mathds{1}\{x_{aij}'b < x_{\tilde{a}ij}'b\} \mathds{1}\{S_{aij} < S_{\tilde{a}ij}\} \bigg) 
\end{align*}

\begin{lemma} \label{lem:conv}
$SQ_{1, \mathcal{A}_{ij}} (b; \sigma_{\mathcal{A}_{ij}})$ converges uniformly almost surely to $Q_{1, \mathcal{A}_{ij}} (b; \sigma_{\mathcal{A}_{ij}})$.
\end{lemma}

\begin{proof}
From the definitions, we have
\begin{align*}
 SQ_{1, \mathcal{A}_{ij}} (b; \sigma_{\mathcal{A}_{ij}}) - Q_{1, \mathcal{A}_{ij}} (b; \sigma_{\mathcal{A}_{ij}})
  \quad \quad \quad \quad \quad \quad \quad \quad \quad \quad \quad \quad  & \\
= {\abs{\mathcal{A}_{ij}} \choose 2}^{-1} \sum_{a \neq \tilde{a} \in \mathcal{A}_{ij}} 
 \bigg[ K \bigg(  \scriptsize  \frac{  \norm{ z_a - z_{\tilde{a}}}}  {\sigma_{\mathcal{A}_{ij}}} \bigg) - \mathds{1}\{z_{a} = z_{\tilde{a}}\} \bigg] 
	 \bigg[  & \mathds{1}\{x_{aij}'b > x_{\tilde{a}ij}'b\} \mathds{1}\{S_{aij} > S_{\tilde{a}ij}\}  \\
			+ & \mathds{1}\{x_{aij}'b < x_{\tilde{a}ij}'b\} \mathds{1}\{S_{aij} < S_{\tilde{a}ij}\} \bigg] 
\end{align*}
Let $\mathcal{B}_{\rho} = \{ b: \abs{b^1} \geq \rho, \, \norm{b} = 1 \}$. Since the term in the second brackets is either 0 or 1, we have 
\begin{align*}
\sup_{b \in \mathcal{B}_{\rho}} \,  \abs{SQ_{1, \mathcal{A}_{ij}} (b; \sigma_{\mathcal{A}_{ij}}) - Q_{1, \mathcal{A}_{ij}} (b; \sigma_{\mathcal{A}_{ij}})}
& \leq {\abs{\mathcal{A}_{ij}} \choose 2}^{-1} \sum_{a \neq \tilde{a} \in \mathcal{A}_{ij}} 
 \abs{ K \bigg(  \scriptsize  \frac{ \norm{ z_a - z_{\tilde{a}}}}  {\sigma_{\mathcal{A}_{ij}}} \bigg) - \mathds{1}\{z_{a} = z_{\tilde{a}}\} } 
\end{align*}
For any $\alpha > 0$, define 
\begin{align*}
C_{1,\mathcal{A}_{ij}} (\alpha) & = {\abs{\mathcal{A}_{ij}} \choose 2}^{-1} \sum_{a \neq \tilde{a} \in \mathcal{A}_{ij}} 
  \abs{ K \bigg(  \scriptsize  \frac{ \norm{ z_a - z_{\tilde{a}}}}  {\sigma_{\mathcal{A}_{ij}}} \bigg) - \mathds{1}\{z_{a} = z_{\tilde{a}}\} } 
   \mathds{1}\{ \norm{z_{a} - z_{\tilde{a}}} \geq \alpha \} \\
C_{2,\mathcal{A}_{ij}} (\alpha) & = {\abs{\mathcal{A}_{ij}} \choose 2}^{-1} \sum_{a \neq \tilde{a} \in \mathcal{A}_{ij}} 
  \abs{ K \bigg(  \scriptsize  \frac{ \norm{ z_a - z_{\tilde{a}}}}  {\sigma_{\mathcal{A}_{ij}}} \bigg) - \mathds{1}\{z_{a} = z_{\tilde{a}}\} } 
   \mathds{1}\{ \norm{z_{a} - z_{\tilde{a}}} < \alpha \} 
\end{align*}
Then we know that 
\begin{align*}
\sup_{b \in \mathcal{B}_{\rho}} \,  \abs{SQ_{1, \mathcal{A}_{ij}} (b; \sigma_{\mathcal{A}_{ij}}) - Q_{1, \mathcal{A}_{ij}} (b; \sigma_{\mathcal{A}_{ij}})} \leq C_{1,\mathcal{A}_{ij}} (\alpha) + C_{2,\mathcal{A}_{ij}} (\alpha) 
\end{align*}
First, let's try to bound $C_{1,\mathcal{A}_{ij}} (\alpha)$. When $\norm{z_{a} - z_{\tilde{a}}} \geq \alpha$, assumption \ref{as:K} and \ref{as:bandwidth} imply that $ K \left(  \scriptsize  \frac{ \norm{ z_a - z_{\tilde{a}}}}  {\sigma_{\mathcal{A}_{ij}}} \right) \rightarrow 0$ as $\abs{\mathcal{A}_{ij}} \rightarrow \infty$. Moreover, we have $ \mathds{1}\{z_{a} = z_{\tilde{a}}\}=0$, so $C_{1,\mathcal{A}_{ij}} (\alpha) \rightarrow 0$ as $\abs{\mathcal{A}_{ij}} \rightarrow \infty$ for any $\alpha >0$ and $b \in \mathcal{B}_{\rho}$.

Now let's try to bound $C_{2,\mathcal{A}_{ij}} (\alpha)$. Since $K$ is bounded from assumption \ref{as:K}, we can find a finite $M$ such that
\begin{align}
C_{2,\mathcal{A}_{ij}} (\alpha) & \leq M  {\abs{\mathcal{A}_{ij}} \choose 2}^{-1} \sum_{a \neq \tilde{a} \in \mathcal{A}_{ij}}  \mathds{1}\{ \norm{z_{a} - z_{\tilde{a}}} < \alpha \} \\
& \leq M {\abs{\mathcal{A}_{ij}} \choose 2}^{-1} \sum_{a \neq \tilde{a} \in \mathcal{A}_{ij}}  \mathds{1}\{ \abs{z_{ai}^{1} - z_{\tilde{a}i}^{1}} < \alpha \} \label{eq:C2}
\end{align}
As $\alpha \rightarrow 0$, equation \eqref{eq:C2} becomes $M  {\abs{\mathcal{A}_{ij}} \choose 2}^{-1} \sum_{a \neq \tilde{a} \in \mathcal{A}_{ij}}  \mathds{1}\{ z_{ai}^{1} - z_{\tilde{a}i}^{1} =0 \}$. It is straightforward to show that assumption \ref{as:support} (b) implies ${\abs{\mathcal{A}_{ij}} \choose 2}^{-1} \sum_{a \neq \tilde{a} \in \mathcal{A}_{ij}}  \mathds{1}\{ z_{ai}^{1} - z_{\tilde{a}i}^{1} =0 \} \rightarrow 0$ as $\abs{\mathcal{A}_{ij}} \rightarrow \infty$. Thus, $C_{2,\mathcal{A}_{ij}} (\alpha)$ can be made arbitrarily small by picking a sufficiently small $\alpha$ and sufficiently large $\abs{\mathcal{A}_{ij}}$. This implies that
\begin{align*}
\lim_{\abs{\mathcal{A}_{ij}} \rightarrow \infty} \: \sup_{b \in \mathcal{B}_{\rho}} \,  \abs{SQ_{1, \mathcal{A}_{ij}} (b; \sigma_{\mathcal{A}_{ij}}) - Q_{1, \mathcal{A}_{ij}} (b; \sigma_{\mathcal{A}_{ij}})} = 0
\end{align*}
\end{proof}

Let $SQ_{2, \mathcal{A}_{ij}} (m; \sigma_{\mathcal{A}_{ij}})$ and   $SQ_{3, \mathcal{A}_{ij}} (m; \sigma_{\mathcal{A}_{ij}})$ be defined analogously to $SQ_{1, \mathcal{A}_{ij}} (b; \sigma_{\mathcal{A}_{ij}})$. Then we have
\begin{lemma} \label{lem:conv2}
$SQ_{2, \mathcal{A}_{ij}} (m; \sigma_{\mathcal{A}_{ij}}) + SQ_{3, \mathcal{A}_{ij}} (m; \sigma_{\mathcal{A}_{ij}})$ converges uniformly almost surely to \\ $Q_{2, \mathcal{A}_{ij}} (m; \sigma_{\mathcal{A}_{ij}}) + Q_{3, \mathcal{A}_{ij}} (m; \sigma_{\mathcal{A}_{ij}})$.
\end{lemma}
\begin{proof}
Similar arguments as in proof of lemma \ref{lem:conv}.
\end{proof}

Combining lemma \ref{lem:conv} and \ref{lem:conv2}, we have verified that condition (iv) holds.

\section{More Tables} \label{sec:tables}

\begin{table}
\begin{threeparttable}[H]
\caption{Summary Statistics for All Destinations, \\ Random Ranking vs. Non-Random Ranking}
\label{tab:sumstat788}
\begin{tabular}{lcccccc}
\hline \hline
\textit{Random ranking sample} &  \\
& Mean & SD & Median & Min & Max & Obs \\
\textit{Hotel level} &  \\
\quad \text{Price} (\$100)&1.72&1.14&1.41&0.10&10.00&1,357,106 \\
\quad \text{Star rating} &3.34&0.89&3.00&1.00&5.00&1,357,106 \\
\quad \text{Review score} &3.81&0.97&4.00&0.00&5.00&1,357,106 \\
\quad \text{Chain} &0.62&0.48&1.00&0.00&1.00&1,357,106 \\
\quad \text{Location score} &3.26&1.53&3.33&0.00&6.98&1,357,106 \\
\quad \text{Promotion} &0.24&0.43&0.00&0.00&1.00&1,357,106  \\
\quad \text{Position} &17.40&10.48&16.00&1.00&40.00&1,357,106 \\
\textit{Impression level} &  \\
\quad \text{Number of hotels displayed} &26.35&8.46&31.00&5.00&38.00&51,510 \\
\quad \text{Booking window (days)} &53.67&62.49&31.00&0.00&498.00&51,510 \\
\quad \text{Click} &1.14&0.66&1.00&1.00&25.00&51,510 \\
\quad \text{Purchase} &0.08&0.27&0.00&0.00&1.00&51,510 \\ \hline 
\textit{Non-Random ranking sample} &  \\
& Mean & SD & Median & Min & Max & Obs \\
\textit{Hotel level} &  \\
\quad \text{Price} (\$100) &1.55&0.97&1.29&0.10&10.00&3,145,937  \\
\quad \text{Star rating} &3.31&0.87&3.00&1.00&5.00&3,145,937  \\
\quad \text{Review score}  &3.92&0.80&4.00&0.00&5.00&3,145,937  \\
\quad \text{Chain} &0.68&0.47&1.00&0.00&1.00&3,145,937  \\
\quad \text{Location score} &3.13&1.50&3.04&0.00&6.98&3,145,937  \\
\quad \text{Promotion} &0.26&0.44&0.00&0.00&1.00&3,145,937  \\
\quad \text{Position} &17.74&10.56&18.00&1.00&40.00&3,145,937  \\
\textit{Impression level} &  \\
\quad \text{Number of hotels displayed} &27.47&7.91&31.00&5.00&38.00&114,526 \\
\quad \text{Booking window (days)} &32.78&48.16&14.00&0.00&482.00&114,526  \\
\quad \text{Click} &1.11&0.59&1.00&1.00&24.00&114,526 \\
\quad \text{Purchase} &0.92&0.28&1.00&0.00&1.00&114,526 \\ \hline \hline
\end{tabular}
\begin{tablenotes}[normal,flushleft]
\item \textit{Notes} This table shows the summary statistics of all of the 788 destinations of the Expedia data set, for both the random ranking sample and the non-random ranking sample. 
\end{tablenotes}
\end{threeparttable}
\end{table}

\end{appendices}

\end{document}